%% file: main.tex
\documentclass[journal]{IEEEtran}
\ifCLASSINFOpdf
\else
\fi
\ifCLASSOPTIONcompsoc
 \usepackage[caption=false,font=normalsize,labelfont=sf,textfont=sf]{subfig}
\else
 \usepackage[caption=false,font=footnotesize]{subfig}
\fi

\input{input.tex}

\usepackage{graphicx}
\usepackage{amsmath}
\usepackage{amsthm}
\usepackage{mathtools}
\usepackage{epsfig}
\usepackage{float}
\usepackage{lipsum}
\usepackage{stfloats}
\usepackage{array}
\usepackage{amssymb}
\usepackage{cite}
\usepackage{url}
\usepackage{algorithm}
\usepackage{algpseudocode}
\usepackage{multirow}
\usepackage{fancyh}
\usepackage{flushend}
\usepackage{upgreek}
\usepackage{dsfont}
\usepackage{gensymb}
\usepackage{booktabs}
\usepackage{bm}
\usepackage[shortlabels]{enumitem}
\usepackage{threeparttable}
\usepackage{physics}
\usepackage{nicematrix}

\usepackage{tikz}
\usepackage{pgfmath, pgfplots}
\usetikzlibrary{3d}

\usepackage{pifont}
\newcommand{\cmark}{\ding{51}}%
\newcommand{\xmark}{\ding{55}}%

\theoremstyle{remark}
\newtheorem{theorem}{Theorem}
\newtheorem{lemma}{Lemma}

\newtheorem{definition}{Definition}

\newtheorem{example}{Example}

\DeclareMathOperator*{\col}{col}
\DeclareMathOperator*{\vspan}{span} 
\DeclareMathOperator*{\supp}{supp} 

\normalsize
\begin{document}
%
\title{Cramer-Rao Bound for Arbitrarily Constrained Sets
\thanks{H. Do and A. Lozano are with Univ. Pompeu Fabra, 08018 Barcelona (e-mail:\{heedong.do, angel.lozano\}@upf.edu).
Their work is supported by the Maria de Maeztu Units of Excellence Programme CEX2021-001195-M funded by MICIU/AEI/10.13039/501100011033, and by the Departament de Recerca i Universitats de la Generalitat de Catalunya.}
}

\author{\IEEEauthorblockN{Heedong~Do},
 {\it Member,~IEEE},
 \and
 \IEEEauthorblockN{Angel~Lozano},
{\it Fellow,~IEEE}
\vspace{-4mm}
}
\maketitle



%


\maketitle

\begin{abstract}
This paper presents a Cramer-Rao bound (CRB) for the estimation of parameters confined to an arbitrary set. Unlike existing results that rely on equality or inequality constraints, manifold structures, or the nonsingularity of the Fisher information matrix, the derived CRB applies to any constrained set and holds for any estimation bias and any Fisher information matrix. The key geometric object governing the new CRB is the tangent cone to the constraint set, whose span determines how the constraints affect the estimation accuracy. This CRB subsumes, unifies, and generalizes known special cases,
offering an intuitive and broadly applicable framework to characterize the minimum mean-square error of constrained estimators. 
\end{abstract}



%
\IEEEpeerreviewmaketitle

\section{Introduction}

The Cramer-Rao bound (CRB) is a reference benchmark for the mean-square error (MSE) of estimators \cite{kay1993fundamentals}.
It relies on local statistical information, the Fisher information matrix, whereby it is tight in the small-error regime.
Although other bounds that exploit global information can be tighter in other regimes \cite{barankin1949locally, abel1993bound, todros2010general}, these other bounds are computationally demanding, analytically intractable, and they sometimes fail to capture threshold effects in the large-error regime \cite{mueller1995estimation, abbey1996barankin, athley2005threshold, todros2015limitations}. In contrast, the CRB offers a computational and analytical tractability that has led to its wide adoption in signal processing, including:
\begin{itemize}
    \item \textit{Parameter estimation.} Parameter estimation in additive Gaussian noise \cite{slepian1954estimation}, frequency estimation \cite{rife1974single}, polynomial phase estimation \cite{peleg1991cramer}, tensor decomposition \cite{sidiropoulos2017tensor}.
    \item \textit{Array processing.} Far-field source localization \cite{stoica1989music}, near-field source localization \cite{huang1991near}, source localization using sparse arrays \cite{wang2017coarrays}.
    \item \textit{Communication.} Synchronization \cite{schmidl1997robust}, signal-to-noise ratio estimation \cite{pauluzzi2000comparison}, wireless localization \cite{patwari2003relative}, feedforward carrier recovery \cite{pfau2009hardware}, channel estimation \cite{larsen2009performance}, or integrated sensing and communications \cite{xiong2023fundamental}.
    \item \textit{Image processing.} Image registration \cite{robinson2004fundamental}, image super-resolution \cite{robinson2006statistical}, image denoising \cite{chatterjee2019is}.
\end{itemize}

The standard CRB concerns parameters residing in the Euclidean space or an open subset thereof. Extended to sets described by equality and inequality constraints, it is termed \textit{constrained CRB} \cite{gorman1990lower}. 
\nocite{marzetta1993simple, stoica1998cramer, moore2007constrained, benhaim2009constrained, zhiguang2011new, benhaim2010cramer}
The type of such constraints has gradually expanded, and the restrictions on the constrained CRB's applicability have eased over time (see Table \ref{table:prior_art}, where the various restrictions are described on the basis of quantities introduced in subsequent sections). 


This paper further generalizes 
the constrained CRB to any arbitrary set.
The ensuing bound is independent of the parameterization of the set and of the basis choice,
and it subsumes a number of established results. 
Resting on a generalized version of the Schur complement lemma, its proof
is intuitive and insightful.

      
      


\begin{table}
\setlength{\tabcolsep}{7pt}
\centering
\begin{threeparttable}
\caption{Prior Art and its Restrictions}
\label{table:prior_art}
\begin{tabular}{c | l | c c c c c} 
\toprule
& & \multicolumn{5}{c}{Restrictions\tnote{*}}
\\
Reference & Constraint & 1 & 2 & 3 & 4 & 5
\\
\midrule
\cite{gorman1990lower} & inequality & \cmark & -
&
-
& \xmark & \xmark
\\
\cite{marzetta1993simple} & equality & \cmark & -
&
-
& \cmark & \cmark
\\
\cite{stoica1998cramer} & equality & \xmark & \cmark & \cmark & \cmark & \cmark
\\
\cite{moore2007constrained} & equality & \xmark & \cmark & \cmark & \cmark & \cmark
\\
\cite{benhaim2009constrained} & equality & \xmark & \xmark & \cmark & \xmark & \xmark
\\
\cite{zhiguang2011new} & equality & \xmark & \cmark & \cmark & \xmark & \xmark 
\\
\cite{benhaim2010cramer} & locally balanced set & \xmark & \xmark & \cmark & \xmark & \xmark 
\\
\textbf{This paper} & any set & \xmark & \xmark & \xmark & \xmark & \xmark
\\
\bottomrule
\end{tabular}
\begin{tablenotes}
\item[*] Fewer restrictions mean wider applicability. The column headers correspond to these restrictions:
\begin{enumerate}
    \item[1.] Nonsingularity of the Fisher information matrix ($\bJ > \zero$).
    \item[2.] Nonsingularity of the Fisher information matrix projected onto the space spanned by the tangent cone ($\bU^\top\bJ\bU > \zero$). A dash indicates that the restriction is immaterial because there is no dependence on $\bU$.
    \item[3.] Orthonormality of the basis of the space spanned by the tangent cone ($\bU^\top\bU=\bI$). A dash indicates that the restriction is immaterial because there is no dependence on $\bU$.
    \item[4.] Range of the estimator included in the constraint set ($\hat{\btheta}\in\Theta$).
    \item[5.] Unbiasedness of the estimator ($\bb(\btheta)=\zero$ for $\btheta\in\Theta$).
\end{enumerate}
\end{tablenotes}
\end{threeparttable}
\end{table}

The manuscript is organized as follows. 
Sec. \ref{sec:tangent_cone} introduces the notion of tangent cone, by means of which the estimation covariance in a constrained set
is bounded in Sec.~\ref{sec:ccrb}. 
By identifying the best such bound, Sec.~\ref{sec:monotonicity} establishes the constrained CRB as a function of span of the tangent cone to the constrained set. Concrete examples of the tangent cone and its span are put forth in Sec.~\ref{sec:examples} for widely used constraint sets. Finally, the paper concludes in Sec.~\ref{sec:discussion}. 

\begin{figure*}
    \centering
    \subfloat[Crescent-shaped set $\Theta\subset \bbR^2$]{
    \begin{tikzpicture}[x=1, y=1]
    \clip (-5,-75) rectangle (105,75);
    \path[draw=black, fill=blue, fill opacity = 0.3]
        (0,50) .. controls (50,80) and (100,50) .. (100,0)
        .. controls (100,-50) and (50,-80) .. (0,-50)
        .. controls (60,-60) and (70,-20) .. (70,0)
        .. controls (70,20) and (60,60) .. (0,50)
        --cycle;
    \end{tikzpicture}
    }\hspace{2mm}
    \subfloat[Tangent cone at a vertex]{
    \begin{tikzpicture}[x=1, y=1]
    \clip (-5,-75) rectangle (105,75);
    \path[draw=black, fill=blue, fill opacity = 0.3]
        (0,50) .. controls (50,80) and (100,50) .. (100,0)
        .. controls (100,-50) and (50,-80) .. (0,-50)
        .. controls (60,-60) and (70,-20) .. (70,0)
        .. controls (70,20) and (60,60) .. (0,50)
        --cycle;
    \path[draw=black, fill=red, fill opacity = 0.2]
        (0,50) -- (75,95) -- (120,70) -- (0,50) -- cycle;
    \filldraw (0,50) circle (1pt) node[below]{$\btheta$};
    \draw (50,67) node[]{$\cT_\Theta(\btheta)$};
    \end{tikzpicture}
    }\hspace{2mm}
    \subfloat[Tangent cone at an edge]{
    \begin{tikzpicture}[x=1, y=1]
    \clip (-5,-75) rectangle (105,75);
    \path[draw=black, fill=blue, fill opacity = 0.3]
        (0,50) .. controls (50,80) and (100,50) .. (100,0)
        .. controls (100,-50) and (50,-80) .. (0,-50)
        .. controls (60,-60) and (70,-20) .. (70,0)
        .. controls (70,20) and (60,60) .. (0,50)
        --cycle;
    \path[draw=black, fill=red, fill opacity = 0.2]
        (70,80) -- (120,80) -- (120,-80) -- (70,-80) -- cycle;
    \filldraw (70,0) circle (1pt) node[left]{$\btheta$};
    \draw (85,10) node[]{$\cT_\Theta(\btheta)$};
    \end{tikzpicture}
    }\hspace{2mm}
    \subfloat[Tangent cone at an interior point]{
    \begin{tikzpicture}[x=1, y=1]
    \clip (-5,-75) rectangle (105,75);
    \path[draw=black, fill=blue, fill opacity = 0.3]
        (0,50) .. controls (50,80) and (100,50) .. (100,0)
        .. controls (100,-50) and (50,-80) .. (0,-50)
        .. controls (60,-60) and (70,-20) .. (70,0)
        .. controls (70,20) and (60,60) .. (0,50)
        --cycle;
    \path[fill=red, fill opacity = 0.2]
        (-5,70) -- (105,70) -- (105,-70) -- (-5,-70) -- cycle;
    \filldraw (80,20) circle (1pt) node[left]{$\btheta$};
    \draw (20,0) node[]{$\cT_\Theta(\btheta)$};
    \end{tikzpicture}
    }
    \caption{Set $\Theta$ with the tangent cone $\cT_\Theta(\btheta)$ shaded in red for different points $\btheta$.}
    \label{fig:crescent}
\end{figure*}
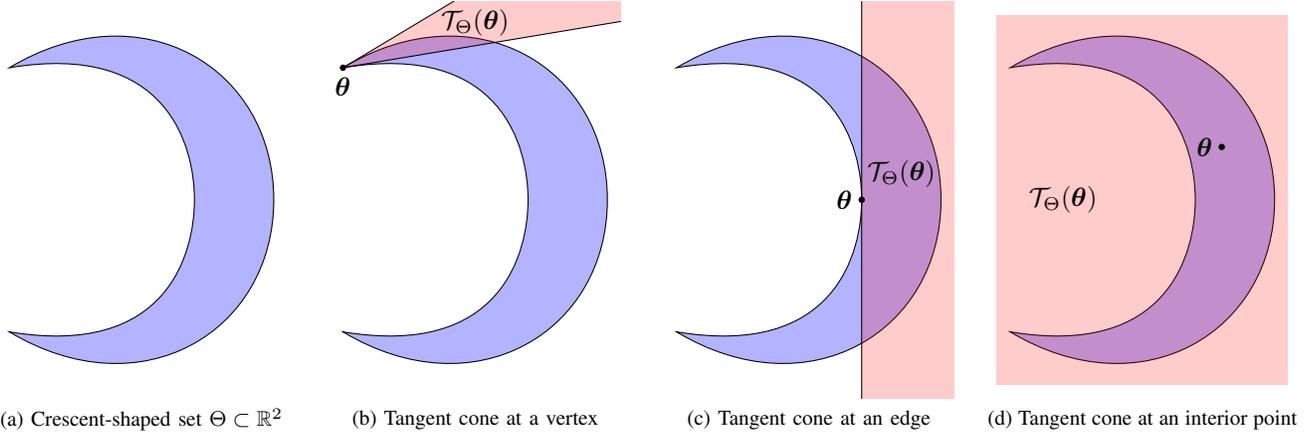

\section{The Tangent Cone}
\label{sec:tangent_cone}

Consider a set $\Theta\subset \bbR^k$ and a function $f:\Theta\rightarrow \bbR^m$, which shall respectively correspond to the constraint set and to either the bias function or the log-likelihood function.
It is assumed that a differentiable extension of $f$
to some open set including $\Theta$ exists.

\begin{definition} \textit{(Tangent vector)}
\label{tv}
A vector $\bu$ is a tangent vector to $\Theta $ at $\btheta \in \Theta$ if there are sequences $\{\btheta_i\in\Theta\}$ and $\{\lambda_i>0\}$ such that $\btheta_i \rightarrow \btheta$ and $\lambda_i \, (\btheta_i-\btheta) \rightarrow \bu$ \cite[Def. 6.1]{rockafellar1998variational}.
\end{definition}
A related notion is that of \textit{feasible direction}. A vector $\bu$ is a feasible direction of $\Theta$ at $\btheta \in \Theta$ if there exists $\epsilon>0$ such that the line
segment connecting $\btheta$ and $\btheta + \epsilon \bu$ lies in $\Theta$. Any feasible direction is a tangent vector, but not vice versa.
Intuitively, tangent vectors are associated with the trajectories over all possible curved paths that stay in the set, while the more restrictive feasible directions correspond to the straight paths that stay in the set.
For instance, on the spherical shell $\Theta = \{\btheta:\|\btheta\|=1\}$, there are plenty of tangent vectors; for $\bu$ perpendicular to $\btheta$, letting
\begin{align}
    \btheta_i = \frac{\btheta + \frac{\bu}{i}}{\|\btheta + \frac{\bu}{i}\|}\in \Theta \qquad \lambda_i = i,
\end{align}
one can easily verify that $\lambda_i(\btheta_i - \btheta) \rightarrow \bu$.
The set of feasible directions, however, is empty at every point on the shell.

The set of all tangent vectors to $\Theta$ at $\btheta$ is called the \textit{tangent cone} to $\Theta$ at $\btheta$, denoted by $\cT_{\Theta}(\btheta)$.
Fig. \ref{fig:crescent} illustrates this concept, which plays a central role in the characterization of the constrained CRB. Concrete examples of the tangent cones are furnished in Sec. \ref{sec:examples}.
Throughout the paper, $\bW$ denotes an arbitrary matrix having $k$ rows whereas $\bU$ is used to indicate a matrix whose columns are in $ \cT_\Theta(\btheta)$.

Also instrumental in the characterization of the constrained CRB is the notion of directional derivative.

\begin{definition} \textit{(Directional derivative)}
The directional derivative of a function $f:\bbR^k \rightarrow \bbR^m$ at $\btheta$ in the direction of $\bu$ is
\begin{align}
    \frac{\partial f}{\partial \bu}\equiv\lim_{t\rightarrow 0} \frac{f(\btheta+t\bu)-f(\btheta)}{t}.
\end{align}
If $f$ is differentiable at $\btheta$, it holds that \cite[Proposition 1.7.14]{hubbard2015vector}
\begin{align}
    \frac{\partial f}{\partial \bu} = \frac{\partial f}{\partial \btheta} \bu,
\end{align}
where $\frac{\partial f}{\partial \btheta} \in \bbR^{m\times k}$ is the derivative of $f$.
\end{definition}

Importantly, the directional derivative of $f$ at $\btheta$ in the direction $\bu\in \cT_\Theta(\btheta)$ does not depend on the choice of the differentiable extension (see App. \ref{Valencia}).

\section{A Constrained Bound on the Covariance Matrix of the Estimator}
\label{sec:ccrb}

Consider the estimation of 
$    \btheta  = [\theta_1 \,\cdots\, \theta_k]
^\top \in \Theta $    
from some observation $\by$ with probability density function $p(\by;\btheta)$. The bias and covariance matrix of the estimator $\hat{\btheta}(\by)$ are
\begin{align}
    \bb\equiv \bbE [\hat{\btheta}]-\btheta
\end{align}    
   and 
    \begin{align}
    \bC_{\btheta} &\equiv\bbE \! \left[ \! \left(\hat{\btheta}-\bbE [\hat{\btheta}] \right) \! \left(\hat{\btheta}-\bbE [\hat{\btheta}] \right)^{\!\!\top} \right] ,
\end{align}
with all expectations taken over $p(\by;\btheta)$.

To procure the MSE matrix, the bias-variance decomposition formula comes in handy, namely \cite[Remark 7]{gorman1990lower}
\begin{align}
    \bbE \! \left[ \big(\hat{\btheta}-\btheta\big) \big(\hat{\btheta}-\btheta \big)^{\! \top} \right]
    &=\bC_{\btheta} + \bb\bb^\top.
\end{align}

In the absence of constraints, i.e., for $\Theta = \bbR^k$, the covariance of the estimator is bounded by the unconstrained CRB, summarized into Thm.~\ref{theorem:unconstrained_crb} below.
The Fisher information matrix is denoted by 
\begin{align}
\bJ \equiv \bbE \! \left[ \! \left(\frac{\partial L}{\partial \btheta} \right)^{\!\!\top} \frac{\partial L}{\partial \btheta} \right],
\end{align}
with $L(\by;\btheta) = \log p(\by;\btheta)$ the log-likelihood function
and with $\col(\cdot)$ returning the column space of a matrix.

\begin{theorem}
\label{theorem:unconstrained_crb}
An estimator 
with bias $\bb$ satisfies
\begin{align}
    \bC_{\btheta} \geq \bigg(\bI + \frac{\partial \bb}{\partial \btheta}\bigg)\bJ ^{\dagger}\bigg(\bI + \frac{\partial \bb}{\partial \btheta}\bigg)^{\!\!\top} \label{unconstrained_crb}
\end{align}
and
\begin{align}
    \col \! \left( \! \bigg(\bI + \frac{\partial \bb}{\partial \btheta}\bigg)^{\!\!\top}\right) \subset \col \bJ . \label{range_crb}
\end{align}
Eq. \eqref{unconstrained_crb} holds with equality if and only if
\begin{align}
    \hat{\btheta}-\bbE[\hat{\btheta}] = \bigg(\bI + \frac{\partial \bb}{\partial \btheta}\bigg)\bJ^\dagger\bigg(\frac{\partial L}{\partial \btheta}\bigg)^{\!\!\top} 
\end{align}
in the mean-square sense.
\end{theorem}
\begin{proof}
A standard proof is provided in App. \ref{app:unconstrained_bound} by means of the Schur complement lemma, itself proved in App.~\ref{app:schur}. 
\end{proof}

Let us now consider a constrained set, $\Theta \subset \bbR^k$.
The key difference when dealing with such a set 
is that the gradients are to be replaced by directional derivatives
that can only be taken in directions corresponding to vectors tangent to the set.

\begin{theorem} 
\label{theorem:crbs}
Given any matrix
\begin{align}
    \bU = \begin{bmatrix}
        \bu_1 & \bu_2 & \cdots & \bu_\ell
    \end{bmatrix},
\end{align}
whose columns are in $\cT_\Theta(\btheta)$,
it holds that 
\begin{align}
    \bC_{\btheta} &\geq \bigg(\bI + \frac{\partial \bb}{\partial \btheta}\bigg)\bU \big(\bU^\top \bJ \bU \big)^{\!\dagger}\bU^\top\bigg(\bI + \frac{\partial \bb}{\partial \btheta}\bigg)^{\!\! \top}  \label{ccrb_arbitrary}
\end{align}
and
\begin{align}
    \col \! \left(\bU^\top\bigg(\bI + \frac{\partial \bb}{\partial \btheta}\bigg)^{\!\!\top}\right) \subset  \col \! \left(\bU^\top\bJ\bU \right).\label{range_ccrb_arbitrary}
\end{align}
Eq. \eqref{ccrb_arbitrary} holds with equality if and only if
\begin{align}
    \hat{\btheta}-\bbE[\hat{\btheta}] = \bigg(\bI + \frac{\partial \bb}{\partial \btheta}\bigg)\bU \big( \bU^\top\bJ\bU \big)^{\!\dagger} \bU^\top \bigg(\frac{\partial L}{\partial \btheta}\bigg)^{\!\!\top} 
\end{align}
in the mean-square sense.
\end{theorem}
\begin{proof}
See App. \ref{app:constrained_bound}.
\end{proof}

As advanced, while the bias and the likelihood function are defined over $\Theta$, it has been assumed in Thm. \ref{theorem:crbs} that their differentiable extensions to some open set including $\Theta$ exist. (For the unbiased case, this extension trivially exists; otherwise, one may check the applicability of the Whitney's extension theorem \cite{whitney1934analytic}.)

Also worth noting is that,
in \eqref{range_ccrb_arbitrary}, $\col \! \big( \bU^\top\bJ\bU \big)$ can be replaced by $\col \! \big( \bJ^{\scriptscriptstyle 1/2}\bU \big)$  thanks to \cite[Thm. 7.64]{axler2024linear}
\begin{align}
    \col\bA^\top = \col \! \big(\bA^\top\bA \big). \label{axler}    
\end{align}
This property shall come in handy in a number of derivations.



\subsection{Reduction Relative to the Unconstrained CRB}

Recall that $\bW$ denotes an arbitrary matrix having $k$ rows.
For $\bJ>\zero$, the matrix function
\begin{align}
    \bW\mapsto \bW \big(\bW^\top \bJ \bW \big)^{\!\dagger}\bW^\top \label{matrix_function_without_bias}
\end{align}
admits the decomposition
\begin{align}
    &\bW (\bW^\top \bJ \bW)^{\dagger}\bW^\top\\
    &\qquad =\bJ^{-\frac{1}{2}} \cdot\bJ^{\frac{1}{2}}\bW(\bW^\top \bJ \bW)^{\dagger}\bW^\top\bJ^{\frac{1}{2}}\cdot \bJ^{-\frac{1}{2}}  \\
    &\qquad=\bJ^{-\frac{1}{2}} \cdot\bJ^{\frac{1}{2}}\bW(\bJ^{\frac{1}{2}}\bW)^\dagger\cdot \bJ^{-\frac{1}{2}}, \label{ccrb_decomposition}
\end{align}
which holds because
\begin{align}
    (\bA^\top\bA)^\dagger \bA^\top = \bA^\dagger . \label{moore_penrose_property}
\end{align}

In \eqref{ccrb_decomposition}, the middle matrix projects onto $\col (\bJ^{\frac{1}{2}}\bW)$, yielding $\bJ^{\frac{1}{2}}\bW(\bJ^{\frac{1}{2}}\bW)^\dagger \leq \bI$. It follows that
\begin{align}
    \bW(\bW^\top \bJ \bW)^{\dagger}\bW^\top
    \leq \bJ^{-\frac{1}{2}}\cdot \bI \cdot \bJ^{-\frac{1}{2}}
    = \bJ^{-1}. \label{constraints_lower_crb}
\end{align}
As the inequality holds for any matrix with $k$ rows, 
it is applicable to $\bU$ in Thm. \ref{theorem:crbs}, yielding
\begin{align}
    \bU(\bU^\top \bJ \bU)^{\dagger}\bU^\top \label{constraints_lower_crb_}
    \leq \bJ^{-1}.
\end{align}
This confirms that the bound obtained with Thm. \ref{theorem:crbs} is always smaller than the unconstrained CRB \cite[Remark~4]{gorman1990lower}.
(While the bias has been omitted for the sake of compactness, it can be straightforwardly incorporated.)

This reduction with respect to the unconstrained CRB can be explicitly quantified.
Let $\bF$ be a matrix whose column space is the orthogonal complement of that of $\bU$. From the nonsingularity of $\bJ$, the column space of $\bJ^{-\frac{1}{2}}\bF$ is the orthogonal complement of that of $\bJ^{\frac{1}{2}}\bU$. This means that \cite[Cor. 1]{stoica1998cramer}
\begin{align}
    \bJ^{\frac{1}{2}}\bU(\bJ^{\frac{1}{2}}\bU)^\dagger 
    &= \bI - \bJ^{-\frac{1}{2}}\bF(\bJ^{-\frac{1}{2}}\bF)^\dagger\\
    &= \bI - \bJ^{-\frac{1}{2}} \bF (\bF^\top\bJ^{-1}\bF)^\dagger \bF^\top \bJ^{-\frac{1}{2}},
\end{align}
and the constrained bound therefore becomes
\begin{align}
    &\bJ^{-\frac{1}{2}} \big(\bI - \bJ^{-\frac{1}{2}} \bF (\bF^\top\bJ^{-1}\bF)^\dagger \bF^\top \bJ^{-\frac{1}{2}}\big) \bJ^{-\frac{1}{2}} \label{no_bias_term} \\
    & \qquad\qquad\qquad\qquad = \bJ^{-1} -
    \underbrace{\bJ^{-1} \bF (\bF^\top\bJ^{-1}\bF)^\dagger \bF^\top \bJ^{-1}}_{\text{reduction}}. \nonumber
\end{align}
(Again, the bias can be straightforwardly incorporated.)

\subsection{Remarks on the Column Space Condition}

The column space condition in \eqref{range_ccrb_arbitrary} can be interpreted as follows: if a bias function fails to fulfill it for some $\bU$ whose columns are in the tangent cone, an estimator with such bias does not exist and any constrained bound becomes immaterial. 

If $\bJ > \zero$, then \eqref{range_ccrb_arbitrary} always holds irrespective of the bias, hence constrained bounds are always meaningful. Indeed, as $\col\bJ = \bbR^k$, we trivially have that
\begin{align}
    \col \! \left(\bigg(\bI + \frac{\partial \bb}{\partial \btheta}\bigg)^{\!\!\top}\right) \subset \col \bJ
\end{align}
where, applying \eqref{axler}, the right-hand side can be replaced with $\col \bJ^{\scriptscriptstyle 1/2}$; multiplying both sides by $\bU^\top$, one obtains
\begin{align}
    \col \! \left(\bU^\top\bigg(\bI + \frac{\partial \bb}{\partial \btheta}\bigg)^{\!\!\top}\right) \subset \col \! \big( \bU^\top\bJ^{\frac{1}{2}} \big),
\end{align}
where the right-hand side is identical to $\col \! \big( \bU^\top\bJ\bU \big)$.

If $\bJ$ is singular, however, one would in principle have to check \eqref{range_ccrb_arbitrary} for every $\bU$ to determine whether the constrained bound is meaningful for some given bias.
The following theorem obviates the need for testing every $\bU$. It turns out that it suffices to test a
unique $\bPi$, which is the projection matrix onto $\vspan \cT_\Theta(\btheta)$, the subspace spanned by the vectors in the cone.
(For a manifold, the tangent cone and its span are identical as $\cT_\Theta(\btheta)$ itself forms a vector space. However, in other cases they are different.)

\begin{theorem}
\label{theorem:range_equivalence}
The following are equivalent:
\begin{enumerate}
    \item[(A)] The column space condition \eqref{range_ccrb_arbitrary} holds for every $\bU$ whose columns are in $\cT_\Theta(\btheta)$.
    \item[(B)] It holds that \begin{align}
        \col \! \left(\bPi \, \bigg(\bI + \frac{\partial \bb}{\partial \btheta}\bigg)^{\!\!\top}\right) \subset \col ( \bPi\bJ\bPi ). \label{range_ccrb_projection}
        \end{align}
    \item [(C)] The column space condition \eqref{range_ccrb_arbitrary} holds for every $\bU$ whose columns are in $\vspan \cT_\Theta(\btheta)$.
\end{enumerate}
\end{theorem}

\begin{proof}
See App. \ref{ICREA}.
\end{proof}

Also noteworthy is that, when the estimator is unbiased, \eqref{range_ccrb_arbitrary} becomes
\begin{align}
    \col \bU^\top \subset \col \! \big( \bU^\top\bJ\bU \big) \label{range_ccrb_full_rank}
\end{align}
and, as the inclusion in the reverse direction holds from
\begin{align}
     \col \! \big( \bU^\top\bJ\bU \big) =  \col \! \big( \bU^\top\bJ^{\frac{1}{2}} \big) \subset \col \bU^\top,
\end{align}
the two column spaces are found to be equal.
Under the proviso that $\bU$ is of full column rank, \eqref{range_ccrb_arbitrary} then reduces, for unbiased estimators, to the nonsingularity of $\bU^\top\bJ\bU$. This had been proved in \cite[Proposition 1]{stoica1998cramer} for the more restrictive case of $\bU$ having orthonormal columns.


\subsection{Remarks on Transformation of Parameters}
Consider
$
    \Theta = \{\btheta = g(\brho): \brho\in\bbR^d\}
$
where $g:\bbR^d \rightarrow \bbR^k$ is a differentiable function.
Letting $\bU = \frac{\partial \btheta}{\partial \brho}$, a constrained bound for $\btheta$ equals the unconstrained bound for $\brho$, namely
\begin{align}
    \frac{\partial \btheta}{\partial \brho}\Bigg( \! \bigg(\frac{\partial \btheta}{\partial \brho}\bigg)^{\!\!\top}\bJ_{\btheta}\frac{\partial \btheta}{\partial \brho}\Bigg)^{\!\dagger}\bigg(\frac{\partial \btheta}{\partial \brho}\bigg)^{\!\!\top}=\frac{\partial \btheta}{\partial \brho}\bJ_{\brho}^\dagger\bigg(\frac{\partial \btheta}{\partial \brho}\bigg)^{\!\!\top},
\end{align}
where the equality follows by applying the chain rule to
\begin{align}
    \bigg(\frac{\partial \btheta}{\partial \brho}\bigg)^{\!\!\top}\bJ_{\btheta}\frac{\partial \btheta}{\partial \brho} 
    &= \bigg(\frac{\partial \btheta}{\partial \brho}\bigg)^{\!\!\top}\bbE \! \left[\bigg(\frac{\partial L}{\partial \btheta}\bigg)^{\!\!\top} \frac{\partial L}{\partial \btheta}\right]\frac{\partial \btheta}{\partial \brho}  \\
    &= \bbE \! \left[\bigg(\frac{\partial L}{\partial \btheta}\frac{\partial \btheta}{\partial \brho}\bigg)^{\!\!\top} \frac{\partial L}{\partial \btheta}\frac{\partial \btheta}{\partial \brho}\right]  \\
    &= \bbE \! \left[ \bigg(\frac{\partial L}{\partial \brho}\bigg)^{\!\!\top} \frac{\partial L}{\partial \brho}\right]  \\
    &= \bJ_{\brho}.
\end{align}
This is consistent with the unconstrained formula under transformation in \cite[Ch. 3.8]{kay1993fundamentals} (see also \cite{yingwei2005regularity, huang2014putting}).
The converse is also possible: with care, the constrained bound can be derived from the unconstrained formula under transformation \cite{moore2007constrained}.

\section{The Constrained CRB}
\label{sec:monotonicity}

Thm. \ref{theorem:crbs} holds for any $\bU$ whose columns are in $\cT_\Theta(\btheta)$, but the bound does depend on the choice of $\bU$. This section establishes the best such bound in the sense of the Loewner order. Although the Loewner order is a partial order, our result demonstrates that there exists a single best bound, which is therefore the constrained CRB.
An enabling result in the derivations is the following.

\begin{lemma}
\label{lemma:existence_V}
There exist $\{\bv_1, \ldots, \bv_d\} \subset \cT_\Theta(\btheta)$ that are a basis of $\vspan \cT_\Theta(\btheta)$.
\end{lemma}
The proof of this lemma, including a simple procedure to construct $\{\bv_1, \ldots, \bv_d\}$, is furnished in App. \ref{bombers}.
Then, letting $\bV = [\bv_1\, \cdots \, \bv_d]$, the projection matrix onto $\vspan \cT_\Theta(\btheta)$ can be written as $\bPi = \bV\bV^\dagger$.

\subsection{Nonsingular Fisher Information Matrix}


The decomposition in \eqref{ccrb_decomposition}, valid for nonsingular $\bJ$,
has several consequences:
\begin{enumerate}
    \item The matrix function in \eqref{matrix_function_without_bias} depends only on the column space of $\bW$, and not on its detailed structure. The middle matrix projects onto $\col (\bJ^{\frac{1}{2}}\bW)$ and the other two matrices do not depend on $\bW$.
    This means that it depends only on 
    \begin{align}
        \col ( \bJ^{\frac{1}{2}}\bW) = \bJ^{\frac{1}{2}} \col \bW,
    \end{align}
    which is a function of $\col \bW$. An alternative derivation of this property is given in \cite[Lemma 5]{gorman1990lower}.
    \item The matrix function in \eqref{matrix_function_without_bias} is monotonically increasing in $\col \bW$ (with respect to the set inclusion). Indeed, if $\col \bW_1 \subset \col \bW_2$, it holds that $\bJ^{\frac{1}{2}}\bW_1(\bJ^{\frac{1}{2}}\bW_1)^\dagger \leq \bJ^{\frac{1}{2}}\bW_2(\bJ^{\frac{1}{2}}\bW_2)^\dagger$ and thus
    \begin{align}
        \bW_1(\bW_1^\top \bJ \bW_1)^{\dagger}\bW_1^\top \leq \bW_2(\bW_2^\top \bJ \bW_2)^{\dagger}\bW_2^\top. \label{monotonicity_nonsingular}
    \end{align}
    \item For $\bU$ whose columns are in $\cT_\Theta(\btheta)$ and $\bV$ resulting from Lemma~\ref{lemma:existence_V},
    \begin{align}
        \bU(\bU^\top \bJ \bU)^{\dagger}\bU^\top \leq \bV(\bV^\top \bJ \bV)^{\dagger}\bV^\top
        \label{bombers0}
    \end{align}
    from $\col \bU \subset \vspan \cT_\Theta (\btheta) = \col\bV$. Recalling that the columns of $\bV$ are in $\cT_\Theta(\btheta)$ by construction, the bound obtained from Thm.~\ref{theorem:crbs}
    with $\bU=\bV$ subsumes all other possible bounds, meaning that
    \begin{align}
    \label{Manresa1}
      \bC_{\btheta} \geq
      \bigg(\bI + \frac{\partial \bb}{\partial \btheta}\bigg)\bV(\bV^\top \bJ \bV)^{\dagger}\bV^\top\bigg(\bI + \frac{\partial \bb}{\partial \btheta}\bigg)^{\!\! \top}
    \end{align}
    with the right-hand side being the constrained CRB.
\end{enumerate}    

Since $\col \bPi = \col \bV$, \eqref{Manresa1} can be recast as
    \begin{align}
    \label{Manresa2}
      \bC_{\btheta} \geq  \bigg(\bI + \frac{\partial \bb}{\partial \btheta}\bigg)\bPi(\bPi^\top \bJ \bPi)^{\dagger}\bPi^\top\bigg(\bI + \frac{\partial \bb}{\partial \btheta}\bigg)^{\!\! \top} ,
    \end{align}
dependent only on $\vspan \cT_\Theta (\btheta)$.
As per the foregoing derivation, this constrained CRB requires $\bJ > \zero$, but this condition is lifted later in this section.



\begin{example}
\label{denoising}
    Given the observation
    \begin{align}
    \by = \btheta + \bw, \label{awgn_model}
\end{align}
where $\bw\sim \cN(\zero,\sigma^2\bI)$, consider the task of
estimating $\btheta\in\Theta$ from $\by$, i.e., denoising.
The Fisher information matrix is
\begin{align}
    \bJ = \frac{1}{\sigma^2}\bI,
\end{align}
and the constrained CRB for unbiased estimators is therefore
\begin{align}
\label{TWC}
 \sigma^2\bPi(\bPi^\top\bI\bPi)^{\dagger}\bPi^\top = \sigma^2\bPi.
\end{align}
\end{example}

\begin{example}
Now let
\begin{align}
    \by = g(\brho) + \bw,
\end{align}
where $\brho\in\bbR^d$ is the vector of the parameters of interest and $\bw = [    w_1 \, \cdots \, w_k ]^\top \sim \cN(\zero,\sigma^2\bI)$.
The performance of the estimator $\hat{\brho}$ is quantified by the MSE matrix in relation to the CRB, $\bJ_{\brho}$.
For large $d$, though, it is customary to compare their diagonal entries instead of the whole matrix because:
\begin{enumerate}
    \item There are only $d$ diagonal entries whereas there are $d^2$ matrix entries.
    \item These diagonal entries represents the variances of each component-wise estimator.
    \item Albeit with loss in information by discarding the off-diagonal entries, both matrices are sure to be identical if their diagonal entries are ($\bM = \zero \Leftrightarrow \tr\bM = 0$).
\end{enumerate}

For some problems, say multidimensional polynomial estimation \cite{do2025multidimensional}, $d$ can be so large that even comparing $d$ diagonal entries is overwhelming. Then, the reconstruction MSE, 
\begin{align}
    \bbE\big[\|g(\hat{\brho})-g(\brho)\|^2\big] ,
\end{align}
is a convenient scalar measure and its constrained CRB
for $\Theta = \{\btheta = g(\brho): \brho\in\bbR^d\}$ adopts, recalling \eqref{TWC},
the very simple form
\begin{align}
    \sigma^2 \tr \bPi &= \sigma^2 \dim (\vspan \cT_{\Theta}(\btheta)) .\label{ccrb_denoising}
\end{align}
\end{example}

\subsection{Arbitrary Fisher Information Matrix}
\label{sec:arbitrary_fisher_information_matrix}

When $\bJ$ is not restricted to being nonsingular,
one might expect that an inequality analogous to \eqref{constraints_lower_crb}, specifically
\begin{align}
    \bW(\bW^\top\bJ\bW)^\dagger\bW^\top \leq \bJ^\dagger, \label{constraints_may_not_lower_crb}
\end{align}
would hold for any $\bW$ having $k$ rows. However, if $\bJ$ is singular, one can find counterexamples, say
\begin{align}
    \bW = \begin{bmatrix}
        1 & 0 \\ 0 & 0
    \end{bmatrix}\qquad\quad
    \bJ = \begin{bmatrix}
        1 & 1 \\ 1 & 1
    \end{bmatrix} \label{counterexample}
\end{align}
for which
\eqref{constraints_may_not_lower_crb} becomes
\begin{align}
    \begin{bmatrix}
        1 & 0 \\ 0 & 0
    \end{bmatrix}
    \leq
    \frac{1}{4}\begin{bmatrix}
        1 & 1 \\ 1 & 1
    \end{bmatrix},
\end{align}
which fails to hold (see \cite[p. 5535]{li2012interpretation} for another example). 

Rather, 
what does hold given $\col \bW_1 \subset \col \bW_2$
is \eqref{monotonicity_nonsingular},
only with the condition $\bW_2^\top\bJ\bW_2 >\zero$ in lieu of $\bJ > \zero$ whenever $\bJ$ is singular 
(see App. \ref{DraR}).

For $\bW_1$ and $\bW_2$ satisfying $\col \bW_1 = \col \bW_2$ as well as $\bW_1^\top\bJ\bW_1>\zero$ and $\bW_2^\top\bJ\bW_2>\zero$, then,
\eqref{monotonicity_nonsingular} yields two-sided inequalities, or equivalently
\begin{align}
    \bW_1(\bW_1^\top\bJ\bW_1)^\dagger\bW_1^\top = \bW_2(\bW_2^\top\bJ\bW_2)^\dagger\bW_2^\top. \label{inequality_implies_range_dependency}
\end{align}
This property does not hold without the conditions $\bW_1^\top\bJ\bW_1>\zero$ and $\bW_2^\top\bJ\bW_2>\zero$; for instance, for
\begin{align}
    \bJ = \begin{bmatrix}
        1 & 0 \\
        0 & 0 \\
    \end{bmatrix}\quad
    \bW_1 = \begin{bmatrix}
        1 & 0\\
        0 & 1
    \end{bmatrix}\quad
    \bW_2 = \begin{bmatrix}
        1 & 1\\
        0 & 1
    \end{bmatrix}
\end{align}
satisfying $\col \bW_1 = \col \bW_2$, \eqref{inequality_implies_range_dependency} becomes
\begin{align}
    \begin{bmatrix}
        1 & 0 \\
        0 & 0 
    \end{bmatrix} = \begin{bmatrix}
        1 & \frac{1}{2} \\
        \frac{1}{2} & \frac{1}{4}
    \end{bmatrix},
\end{align}
which is not true.
Bottom line, $\bW(\bW^\top\bJ\bW)^\dagger\bW^\top$
depends only on $\col \bW$ provided that $\bW^\top\bJ\bW > \zero$. 

The above suggests a natural greedy algorithm to find the constrained CRB by appending columns to $\bU$ until $\bU^\top\bJ\bU$ becomes singular. However, there is still no guarantee that the resulting bound is the best one as there are infinitely many greedy paths. The following theorem gives a definitive answer to the best possible bound. This becomes possible by explicitly re-introducing the bias term into \eqref{monotonicity_nonsingular}, and the column space condition in \eqref{range_ccrb_arbitrary} plays a central role in establishing this result.

\begin{theorem}
\label{theorem:ccrb_monotonicity}
Provided that \eqref{range_ccrb_projection} holds,
the function
\begin{align}
    \bW \mapsto \bigg(\bI + \frac{\partial \bb}{\partial \btheta}\bigg)\bW(\bW^\top \bJ \bW)^{\dagger}\bW^\top\bigg(\bI + \frac{\partial \bb}{\partial \btheta}\bigg)^{\!\! \top} \label{matrix_function_with_bias}
\end{align}
whose domain is the set of matrices with $\col \bW =\vspan \cT_\Theta(\btheta)$, is monotonic with respect to $\col\bW$, i.e.,
\begin{align}
    &\bigg(\bI + \frac{\partial \bb}{\partial \btheta}\bigg)\bW_1(\bW_1^\top \bJ \bW_1)^{\dagger}\bW_1^\top\bigg(\bI + \frac{\partial \bb}{\partial \btheta}\bigg)^{\!\! \top} \label{ccrb_monotonicity}\\
    &\qquad\qquad\quad \leq \bigg(\bI + \frac{\partial \bb}{\partial \btheta}\bigg)\bW_2(\bW_2^\top \bJ \bW_2)^{\dagger}\bW_2^\top\bigg(\bI + \frac{\partial \bb}{\partial \btheta}\bigg)^{\!\!\top} \nonumber
\end{align}
if $\col \bW_1 \subset \col \bW_2$.
\end{theorem}
\begin{proof}
See App. \ref{SOMMA}
\end{proof}



The consequences of Thm. \ref{theorem:ccrb_monotonicity} are as follows:
\begin{enumerate}
\item For $\bW_1$ and $\bW_2$ satisfying $\col \bW_1 = \col \bW_2 \subset \vspan \cT_\Theta(\btheta)$, the theorem gives two-sided inequalities. This implies that \eqref{matrix_function_with_bias} depends on $\bW$ only via $\col \bW$.
\item For $\bU$ whose columns are in $\cT_\Theta(\btheta)$ and $\bV$ resulting from Lemma~\ref{lemma:existence_V},
\begin{align}
    &\bigg(\bI + \frac{\partial \bb}{\partial \btheta}\bigg)\bU(\bU^\top \bJ \bU)^{\dagger}\bU^\top\bigg(\bI + \frac{\partial \bb}{\partial \btheta}\bigg)^{\!\!\top}\\
    &\qquad\qquad\leq \bigg(\bI + \frac{\partial \bb}{\partial \btheta}\bigg)\bV(\bV^\top \bJ \bV)^{\dagger}\bV^\top\bigg(\bI + \frac{\partial \bb}{\partial \btheta}\bigg)^{\!\!\top} \nonumber
\end{align}
from $\col \bU \subset \vspan \cT_\Theta (\btheta) = \col\bV$. This implies that the constrained CRB is given by \eqref{Manresa1} and \eqref{Manresa2} even if $\bJ$ is singular.

\item When $\vspan \cT_\Theta(\btheta) = \bbR^k$, the constrained CRB boils down to its unconstrained counterpart. This generalizes the result in \cite{gorman1990lower}, which proves that the two bounds coincide if no equality constraints are active.
\item The constrained CRB is always smaller than its unconstrained counterpart if the latter is meaningful, i.e., if $\col \! \big( (\bI + \frac{\partial \bb}{\partial \btheta})^{\!\top}\big) \subset \col \bJ$.  

\end{enumerate}

\section{Some Constraint Sets and Their Tangent Cones}
\label{sec:examples}

For a broad class of sets, the tangent cones have been characterized. This section reviews established results and pinpoints their connections with the constrained CRB literature.

\begin{table*}
\setlength{\tabcolsep}{3pt}
\addtolength{\leftskip} {-2cm}
\addtolength{\rightskip} {-2cm}
\centering
\begin{threeparttable}
\caption{Tangent Space of Some Smooth Manifolds}
\label{table:manifolds}
\begin{tabular}{ll|l|l} 
\toprule
Manifold & & Tangent space & Reference \\
\midrule
Sphere & $\{\btheta\in\bbR^k: \|\btheta\|=1\}$ & $\{\bu:\btheta^\top\bu = 0\}$ & \cite[Eq. 3.3]{boumal2023introduction}
\\
Stiefel manifold & $\{\bX\in\bbR^{p\times q}: \bX^\top\bX = \bI\}$ & $\{\bU: \bX^\top\bU + \bU^\top\bX = \zero\}$ &
\cite[Eq. 7.17]{boumal2023introduction}
\\
Orthogonal group & $\{\bX\in\bbR^{p\times p}: \bX^\top\bX = \bI\}$ & $\{\bU: \bX^\top\bU + \bU^\top\bX = \zero\}$ & Special case 
\\
Special orthogonal group & $\{\bX\in\bbR^{p\times p}: \bX^\top\bX = \bI, \det \bX = 1\}$ & $\{\bU: \bX^\top\bU + \bU^\top\bX = \zero\}$ & \cite[Eq. 7.37]{boumal2023introduction}
\\
Fixed-rank matrices\tnote{*} & $\{\bX\in\bbR^{p\times q}: \rank \bX = r\}$ & $\bigg\{\begin{bmatrix}
     \bU_r & \bU_{p-r}
 \end{bmatrix} \begin{bmatrix}
     * & *\\
     * & \zero
 \end{bmatrix}\begin{bmatrix}
     \bV_r^\top \\ \bV_{q-r}^\top
 \end{bmatrix} \bigg\}$  & 
 \cite[Eq. 7.48]{boumal2023introduction}
\\
Fixed-rank positive-definite matrices\tnote{$\dagger$} & $\{\bX\in\bbR^{p
\times p}:\bX\geq \zero, \rank X = r\}$ & $\bigg\{\begin{bmatrix}
     \bU_r & \bU_{p-r}
 \end{bmatrix} \begin{bmatrix}
     * & *\\
     * & \zero
 \end{bmatrix}\begin{bmatrix}
     \bU_r^\top \\ \bU_{p-r}^\top
 \end{bmatrix} \bigg\}$ & \cite[Prop. 2.2]{vandereycken2009embedded}
\\
Positive-definite matrices & $\{\bX\in\bbR^{p\times p}: \bX > \zero\}$ & $\{\bU: \bU= \bU^\top\}$ & Special case
\\
\bottomrule
\end{tabular}
\begin{tablenotes} \footnotesize
\item[*] $\bU_r$ and $\bV_r$ are the singular vectors corresponding to nonzero singular values, and the entries with asterisks can be arbitrarily chosen.
\item[$\dagger$] $\bU_r$ are the eigenvectors corresponding to nonzero eigenvalues, and the entries with the asterisks must be such that the middle matrix is symmetric.
\end{tablenotes}
\end{threeparttable}
\end{table*}

\subsection{Smooth Manifolds}
A vast majority of previous studies address equality constraints of the form,
\begin{align}
    h(\btheta) = \zero, \label{equality_constraint}
\end{align}
where $h:\bbR^k\rightarrow \bbR^{k-m}$ is a continuously differentiable function with a full-rank gradient matrix. 
Although these studies consider the global constraint in \eqref{equality_constraint}, owing to the  local nature of the constrained CRB
their derivations apply to smooth manifolds, locally represented by \eqref{equality_constraint}.

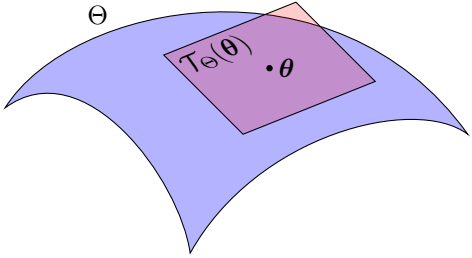
\begin{figure}
    \centering
    \begin{tikzpicture}[x=1, y=1]

    \draw (110,95) node[]{$\Theta$};

    \path[draw=black, fill=blue, fill opacity = 0.3]
        (75,60) .. controls (105,82) and (141,33) .. (145,5)
        .. controls (170,50) and (230,65) .. (250,50)
        .. controls (200,120) and (100,100) .. (75,60)
        --cycle;

    \filldraw[fill=red, fill opacity = 0.2] (185,100) -- (215,70) -- (165,50) -- (135,80) -- cycle ;

    \draw (155,80) node[rotate=22, xslant = -0.5]{$\cT_{\Theta}(\btheta)$};

    \filldraw (175,75) circle (1pt) node[right]{$\btheta$};

    \end{tikzpicture}
    \caption{Tangent cone to $\Theta$ at $\btheta$, which for a smooth manifold is identical to the tangent space.}
    \label{fig:tangent_space}
\end{figure}

A welcome property of manifolds
is that their tangent cone boils down to the tangent space, i.e., the set of all derivatives of smooth curves passing through the point \cite[Example 6.8]{rockafellar1998variational} (see Fig.~\ref{fig:tangent_space}). Put differently, the tangent cone at $\btheta$ is the null space of the gradient matrix, namely
\begin{align}
    \cT_{\Theta}(\btheta) = \bigg\{\bu\in\bbR^k: \frac{\partial h}{\partial \btheta} \bu = \zero\bigg\}. \label{tangent_space_equality}
\end{align}
Table \ref{table:manifolds} lists common smooth manifolds and their tangent spaces.
Notably, it is not straightforward to represent certain manifolds, such as the manifold of fixed-rank matrices, in the form of the equality constraint in \eqref{equality_constraint}.



\begin{example}
Consider the orthogonal group in three dimensions \cite[Sec. IV-B]{chepuri2014rigid},
\begin{align}
    \big\{\bX = \begin{bmatrix}
    \bx_1 & \bx_2 & \bx_3
\end{bmatrix} \in\bbR^3: \bX^\top\bX=\bI\big\},
\end{align}
where, as it is a matrix, the parameter is denoted by $\bX$ rather than $\btheta$.
Vectorization of the constraint, $\bX^\top\bX=\bI$, to put it in the form of \eqref{equality_constraint}, gives
\begin{align}
    \begin{bmatrix}
        \bx_1^\top\bx_1-1\\
        \bx_2^\top\bx_1\\
        \bx_3^\top\bx_1\\
        \bx_2^\top\bx_2-1\\
        \bx_3^\top\bx_2\\
        \bx_3^\top\bx_3-1
    \end{bmatrix}
    = 
    \begin{bmatrix}
        0\\
        0\\
        0\\
        0\\
        0\\
        0
    \end{bmatrix},
\end{align}
where redundant constraints have been removed.
Its gradient matrix is
\begin{align}
    \begin{bmatrix}
        2\bx_1^\top & \zero^\top & \zero^\top\\
        \bx_2^\top & \bx_1^\top & \zero^\top\\
        \bx_3^\top & \zero^\top & \bx_1^\top\\
        \zero^\top & 2\bx_2^\top & \zero^\top\\
        \zero^\top & \bx_3^\top & \bx_2^\top\\
        \zero^\top & \zero^\top & 2\bx_3^\top
    \end{bmatrix}.
\end{align}
Through some tedious computations, one can deduce that
\begin{align}
    \left\{\begin{bmatrix}
        -\bx_2\\
        \bx_1\\
        \zero
    \end{bmatrix},
    \begin{bmatrix}
        -\bx_3\\
        \zero\\
        \bx_1
    \end{bmatrix},
    \begin{bmatrix}
        \zero\\
        -\bx_3\\
        \bx_2
    \end{bmatrix}\right\}
\end{align}
is a basis of the null space of the gradient matrix, i.e., of the tangent space.
Unvectorization gives
\begin{align}
    &\big\{\begin{bmatrix}
    -\bx_2 & \bx_1 & \zero
    \end{bmatrix},
    \begin{bmatrix}
    -\bx_3 & \zero & \bx_1
    \end{bmatrix},
    \begin{bmatrix}
    \zero & -\bx_3 & \bx_2
    \end{bmatrix}\big\}\\
    &=\left\{\bX\begin{bmatrix}
        0 & 1 & 0\\
        -1 & 0 & 0\\
        0 & 0 & 0
    \end{bmatrix}\!\!,
    \bX\begin{bmatrix}
        0 & 0 & 1\\
        0 & 0 & 0\\
        -1 & 0 & 0
    \end{bmatrix}\!\!,
    \bX\begin{bmatrix}
        0 & 0 & 0\\
        0 & 0 & 1\\
        0 & -1 & 0
    \end{bmatrix}\right\}, \nonumber
\end{align}
which is how is presented in Table \ref{table:manifolds}. A similar result for the set of unitary matrices can be found in \cite{jagannatham2004cramer}.
\end{example}

\begin{example}
\label{cold}
Another example, studied in \cite[Sec. III]{tang2011lower}, is the manifold of rank-$r$ matrices\footnote{The set of matrices with rank equal to or less than $r$
is not a manifold.
Nevertheless, its tangent cone has been characterized in \cite{cason2013iterative, schneider2015convergence, olikier2025tangent} as
\begin{align}
    \bigg\{\begin{bmatrix}
     \bU_r & \bU_{p-r}
     \end{bmatrix} \begin{bmatrix}
         * & *\\
         * & \diamond
     \end{bmatrix}\begin{bmatrix}
         \bV_r^\top \\ \bV_{q-r}^\top
     \end{bmatrix} \bigg\}, \nonumber
\end{align}
where the submatrices indicated by asterisks are arbitrary while the submatrix indicated by a diamond has rank equal to or less than $r-\rank \bX$.}
\begin{align}
    \{\bX\in\bbR^{n\times k}: \rank \bX = r\} .
\end{align}
Given the singular value decomposition
\begin{align}
    \bX = \begin{bmatrix}
     \bU_r & \bU_{p-r}
    \end{bmatrix} \begin{bmatrix}
        \bSigma & \zero\\
        \zero & \zero
    \end{bmatrix} \begin{bmatrix}
     \bV_r^\top \\ \bV_{q-r}^\top
    \end{bmatrix},
\end{align}
where $\bU_r$ and $\bV_r$ are the singular vectors corresponding to nonzero singular values whereas $\bU_{p-r}$ and $\bV_{q-r}$ are those corresponding to zero singular values, 
the rank 
is invariant under small perturbations in the direction of
\begin{align}
    &\bU_r\, (\text{any $r\times r$ matrix})\, \bV_r^\top \label{tangent_space_fixed_rank}\\
    &\qquad\text{or }\bU_r\, (\text{any $r\times (q-r)$ matrix})\, \bV_{q-r}^\top \nonumber\\
    &\qquad\text{or }\bU_{p-r}\, (\text{any $(p-r)\times r$ matrix})\, \bV_r^\top. \nonumber
\end{align}
That is, the above are the feasible directions.
This is because the perturbed matrices are
\begin{align}
    &\begin{bmatrix}
     \bU_r & \bU_{p-r}
    \end{bmatrix}\! \begin{bmatrix}
        \bSigma + \epsilon\, (\text{any $r\times r$ matrix}) & \zero\\
        \zero & \zero
    \end{bmatrix}\! \begin{bmatrix}
     \bV_r^\top \\ \bV_{q-r}^\top
    \end{bmatrix} \\
    &\begin{bmatrix}
     \bU_r & \bU_{p-r}
    \end{bmatrix}\! \begin{bmatrix}
        \bSigma & \epsilon\, (\text{any $r\times (q-r)$ matrix})\\
        \zero & \zero
    \end{bmatrix}\! \begin{bmatrix}
     \bV_r^\top \\ \bV_{q-r}^\top
    \end{bmatrix}\\
    &\begin{bmatrix}
     \bU_r & \bU_{p-r}
    \end{bmatrix}\! \begin{bmatrix}
        \bSigma & \zero\\
        \epsilon\, (\text{any $(p-r)\times r$ matrix}) & \zero
    \end{bmatrix}\! \begin{bmatrix}
     \bV_r^\top \\ \bV_{q-r}^\top
    \end{bmatrix},
\end{align}
respectively, where $\epsilon>0$ is a small constant. As the orthogonal matrices on both sides are invertible, the rank of the perturbed matrices is identical to that of the matrices in the middle, which is $r$ for small $\epsilon$.
Altogether, \eqref{tangent_space_fixed_rank} characterizes the feasible directions on which the rank of $\bX$ is preserved, hence the tangent space.


\end{example}

A multidimensional extension of Example \ref{cold} is possible. Unlike for matrices, for tensors several notions of rank exist, each associated with a tensor decomposition. Among these, the set of fixed-rank tensors associated with the Tucker and tensor train decompositions form manifolds. The tangent space is characterized for the Tucker decomposition in \cite[Eq. 2.5]{koch2010dynamical} and for the tensor train one in \cite[Thm.~2]{holtz2012manifolds}.


\subsection{Inequality Constraints}
\label{sec:inequality_constraints}
In \cite{gorman1990lower}, inequality constraints,
\begin{align}
    h(\btheta) \leq \zero, \label{inequality_constraint}
\end{align}
are considered where the inequality is applied in a componentwise manner and $h:\bbR^k\rightarrow \bbR^{k-m}$ is continuously differentiable. As long as the tangent cone has an interior point, its span is $\bbR^k$ and the inequality constraints do not affect the CRB \cite[Lemma 4]{gorman1990lower}.

\begin{figure}
    \centering
    \begin{tikzpicture}[
        scale=1, 
        >=stealth, 
        x={(-0.4 cm,-0.4 cm)}, 
        y={(0.8 cm,-0.1 cm)},
        z={(0 cm,0.8 cm)}
    ]  
        \fill[fill opacity = 0.3, fill = blue] 
            (2,2,0) -- (-2,2,0) -- (-2,-2,0) -- (2,-2,0) -- cycle;
        \fill[fill opacity = 0.3, fill = blue] 
            (0,2,2) -- (0,-2,2) -- (0,-2,-2) -- (0,2,-2) -- cycle;
        \fill[fill opacity = 0.3, fill = blue] 
            (2,0,2) -- (-2,0,2) -- (-2,0,-2) -- (2,0,-2) -- cycle;
        \draw[line width=1pt, ->]
            (-2,0,0) -- (2,0,0) 
            node[above left, yshift = -0.4 cm] {};
        \draw[line width=1pt, ->]
            (0,-2,0) -- (0,2,0)
            node[above left, xshift = 0.3 cm, yshift = 0.1 cm] {};
        \draw[line width=1pt, ->]
            (0,0,-2) -- (0,0,2)
            node[above, xshift = -0.3 cm] {};
    \end{tikzpicture} 
    \caption{The set of sparse vectors for $k=3$ and $s=2$, i.e., $\Theta = \{\btheta\in\bbR^3: \|\btheta\|_0 \leq 2\}$.}
    \label{fig:sparse_set}
\end{figure}


\subsection{Set of Sparse Vectors}
A constrained CRB for the compressed sensing problem is studied in \cite{benhaim2010cramer}, where the associated set of parameters is the set of sparse vectors,
\begin{align}
    \Theta = \{\btheta\in\bbR^k: \|\btheta\|_0 \leq s\}. \label{sparse_set}
\end{align}
It can be verified that its tangent cone is
\begin{align}
    \cT_\Theta (\btheta) = \{\bu: |\supp\bu \cup \supp\btheta| \leq s\}, \label{tangent_cone_sparse_set}
\end{align}
where $\supp$ is the set of the indices whose components are nonzero. Letting $\bee_i$ be the $i$th standard basis vector,
the results in \cite[Thm. 2]{benhaim2010cramer} can be obtained from Thm. \ref{theorem:crbs} via the fact that
\begin{align}
    \{\bee_1, \bee_2, \ldots, \bee_k\} \subset \cT_\Theta(\btheta)
\end{align}
when $|\supp \btheta| < s$ and
\begin{align}
    \cT_\Theta(\btheta) = \{\bu: \supp \bu \subset \supp \btheta\}
\end{align}
when $|\supp \btheta| = s$. Indeed,
\begin{align}
\label{fred}
    \vspan \cT_\Theta(\btheta) = \begin{cases}
        \bbR^k & \text{ if }|\supp \btheta| < s\\
        \{\bu: \supp \bu \subset \supp \btheta\} & \text{ if }|\supp \btheta| = s
    \end{cases}.
\end{align}

\subsection{Remarks}

The Cartesian product naturally appears when there are multiple parameters to estimate \cite{sadler2001bounds, tang2011constrained, zhao2014model, elnakeeb2018line}. 
For smooth manifolds, their Cartesian product is again a smooth manifold and their tangent space is the Cartesian product of their tangent spaces \cite[Prop. 3.20]{boumal2023introduction}. A similar statement holds for tangent cones under a more general setting \cite[Prop. 6.41]{rockafellar1998variational}.

In turn, intersection naturally appears when multiple constraints are imposed \cite{soloveychik2016joint, elnakeeb2018line, meriaux2021matched, halihal2025cramer}.
In such cases, \cite[Thm.~6.42]{rockafellar1998variational} comes in handy.

\section{Summary and Discussion}
\label{sec:discussion}



The CRB on the estimation of a parameter confined to an arbitrary set has been derived. Besides the Fisher information matrix that dictates its unconstrained brethren, the new ingredient in the constrained CRB is the span of the tangent cone to the set at the parameter. The span of the tangent cone therefore fully describes the set as far as the CRB is concerned, and an overview of the tangent cones for a number of common sets as been presented. The constrained CRB is valid for any bias and any Fisher information matrix, singular or not. Conditions have also been provided under which the estimator does not exist, rendering the CRB immaterial.

To finish the paper, and for the sake of completeness, two further aspects that are specific to unbiased constrained estimators are entertained next.

\subsection{Existence of Unbiased Estimators}

A discussion on the existence of unbiased estimators can be found in \cite{somekh2017non}. If the constraint set features an extreme point, an unbiased estimator satisfying $\hat{\btheta}(\by)\in\Theta$ for all $\by$ does not exist under a very mild condition, rendering the CRB immaterial (see App.~\ref{Chelsea}).
A similar conclusion is reached in \cite[Thm. 2]{chepuri2014rigid} for the orthogonal group, the set of all orthogonal matrices. 

The foregoing conclusion critically relies on 
$\hat{\btheta}(\by) \in\Theta$ for all $\by$. One can construct an unbiased constrained estimator by transcending this premise and allowing $\hat{\btheta}$ to live outside the set, even as $\btheta$ is confined to it.

\begin{example}
    Consider denoising on the set 
\begin{align}
    \Theta = \{\btheta\in\bbR^k: \|\btheta\| = 1\}
\end{align}
and consider the estimator
\begin{align}
    \hat{\btheta} = a\frac{\by}{\|\by\|},
\end{align}
which satisfies
\begin{align}
    \bbE\big[\hat{\btheta}\big] &= a \, \bbE\bigg[\frac{\btheta+\bw}{\|\btheta+\bw\|}\bigg]. \label{sphere_example}
\end{align}
It can be seen from sheer symmetry that $\bbE\big[\frac{\btheta+\bw}{\|\btheta+\bw\|}\big]$
is a scalar multiple of $\btheta$, i.e.,
\begin{align}
    \bbE\bigg[\frac{\btheta+\bw}{\|\btheta+\bw\|}\bigg] = \bigg\|\bbE\bigg[\frac{\btheta+\bw}{\|\btheta+\bw\|}\bigg]\bigg\|\,\btheta
\end{align}
where the multiplicative factor is independent of $\btheta$.
Setting $a$ to equal the reciprocal of this scalar multiplication factor gives an unbiased estimator. 
\end{example}


Another workaround is the modification of the notion of unbiasedness; refer to \cite{smith2005covariance, boumal2013intrinsic, nitzan2019cramer}.

\subsection{A Limitation of the Unbiased Constrained CRB}

For a growing number of observations, which in additive-noise models corresponds to a high signal-to-noise ratio, it is generally believed that the bias of nice estimators---those approaching the CRB---tends to vanish and the unbiased CRB becomes fully pertinent.
A theoretical underpinning of this believe is that the maximum-likelihood estimator is asymptotically normal with zero mean and variance identical to the CRB \cite[Thm.~7.1]{kay1993fundamentals}. This has been argued for equality constraints \cite{moore2008maximum}. As evidenced next, though, it does not apply when the constraint set is not smooth.


\begin{example}
\label{flick}
    Consider denoising on the set
\begin{align}
    \Theta = \{\btheta\in\bbR^k: \|\btheta\|_0 \leq 1\}, \label{one_sparse_set}
\end{align}
where the origin is a corner point; this problem has been studied in \cite{jung2011unbiased}. The maximum-likelihood estimator selects the component of $\by$ with the largest absolute value and zeroes out the remaining components. Its MSE at $\btheta = \zero$ is then
\begin{align}
    \bbE\Big[\max_{\ell=1,\ldots,k} w_\ell^2\Big] \sim 2\sigma^2 \log k
\end{align}
as $k\rightarrow \infty$ \cite{iosif2019expected}; this is \emph{lower} than the unbiased constrained CRB, which equals $\sigma^2 k$.
\end{example}

The crux of the matter is that, while the bias does vanish, its gradient 
does not vanish at a corner.
In the case of Example~\ref{flick}, from the invariance of \eqref{one_sparse_set} under scaling it holds that
\begin{align}
    \bb(\btheta,\sigma) = \sigma\bb\big(\tfrac{\btheta}{\sigma},1\big)
\end{align}
and, because of symmetry, the bias at $\btheta$ is proportional to $\btheta$. Altogether, for $\btheta = \theta\bee_\ell$, one can write
\begin{align}
    \bb(\btheta,\sigma) = \sigma u\big(\tfrac{\theta}{\sigma}\big)\bee_\ell
\end{align}
for some function $u: \bbR\rightarrow \bbR$. Although the bias vanishes everywhere as $\sigma\rightarrow 0$, the bias gradient at the origin lingers:
\begin{align}
    \frac{\partial \bb}{\partial \btheta} = u'(0) \bI.
\end{align}
At corner points, therefore, the unbiased constrained CRB is not pertinent.



\section*{Acknowledgment}

Valuable feedback from Andrea Pizzo (UPF) and Stefano Fortunati (Telecom SudParis) is gratefully acknowledged.






\appendices






\section{}
\label{Valencia}

\begin{theorem}
For a differentiable function $f:\bbR^k\rightarrow \bbR^m$ and a set $\Theta$, the directional derivative of $f$ with respect to a direction $\bu\in\cT_\Theta(\btheta)$ at $\btheta$, denoted by
\begin{align}
    \frac{\partial f}{\partial \bu} = \frac{\partial f}{\partial \btheta} \bu,
\end{align}
is completely determined by the restriction of $f$ to $\Theta$.
\end{theorem}
\begin{proof}
From $\bu \in \cT_\Theta(\btheta)$, there is a sequence $\{\btheta_i\in\Theta\}$ and a sequence $\{\lambda_i>0\}$ with $\btheta_i \rightarrow \btheta$ and $\lambda_i \, (\btheta_i-\btheta) \rightarrow \bu$. From the differentiability,
\begin{align}
    f(\btheta_i)-f(\btheta) = \frac{\partial f}{\partial \btheta} (\btheta_i-\btheta) + o(\|\btheta_i-\btheta\|),
\end{align}
which is equivalent to
\begin{align}
    \frac{\partial f}{\partial \btheta}\cdot \lambda_i ( \btheta_i-\btheta ) = \lambda_i \big( f(\btheta_i)-f(\btheta) \big) + o(\|\lambda_i(\btheta_i-\btheta)\|).
\end{align}
Letting $i\rightarrow \infty$, we have that
\begin{align}
    \frac{\partial f}{\partial \bu} = \lim_{i\rightarrow \infty} \lambda_i\big(f(\btheta_i)-f(\btheta)\big),
\end{align}
whose right-hand side is determined by the restriction of $f$ to $\Theta$.
\end{proof}

\section{}
\label{app:schur}


The Schur Complement Lemma is the workhorse of the proofs in Apps. \ref{app:unconstrained_bound} and \ref{app:constrained_bound}.

\begin{lemma} \textit{(Schur complement lemma)}
\label{lemma:schur}
The block matrix
\begin{align}
    \begin{bmatrix}
        \bA &\bB\\
        \bB^\top & \bC
    \end{bmatrix} \label{block_matrix}
\end{align}
is positive-semidefinite
if and only if
\begin{align}
    \bC & \geq\zero \label{schur_condition1}  \\
    \bA - \bB\bC^\dagger\bB^\top & \geq \zero  \label{schur_condition2}  \\
    (\bI-\bC\bC^\dagger)\bB^\top & =\zero, \label{schur_condition3} 
\end{align}
where, since $\bC\bC^\dagger$ projects onto the column space of $\bC$, the last condition is equivalent to
\begin{align}
    \col\bB^\top \subset \col\bC . \label{rangeCondition}
\end{align}
\end{lemma}
\begin{proof}
Although proofs are readily available in many texts
\cite{albert1969conditions, kreindler1972conditions, zhang2006schur},
a self-contained one is provided herein. The `if' part follows from the identity
\begin{align}
    &\begin{bmatrix}
        \bI & -\bB\bC^\dagger\\
        \zero & \bI
    \end{bmatrix}
    \begin{bmatrix}
        \bA &\bB\\
        \bB^\top & \bC
    \end{bmatrix}
    \begin{bmatrix}
        \bI & \zero\\
        -\bC^\dagger\bB^\top & \bI
    \end{bmatrix}\\
    &=\begin{bmatrix}
        \bA-\bB\bC^\dagger\bB^\top-\bB(\bC^\dagger - \bC^\dagger\bC\bC^\dagger)\bB^\top & \bB(\bI-\bC^\dagger\bC) \\
        (\bI-\bC\bC^\dagger)\bB^\top & \bC
    \end{bmatrix} \nonumber\\
    &=\begin{bmatrix}
        \bA-\bB\bC^\dagger\bB^\top & \bB(\bI-\bC^\dagger\bC) \\
        (\bI-\bC\bC^\dagger)\bB^\top & \bC
    \end{bmatrix} \label{block_identity} \\
    & = \begin{bmatrix}
        \bA-\bB\bC^\dagger\bB^\top & \zero \\
        \zero & \bC
    \end{bmatrix}, \label{aclasse}
\end{align}
where \eqref{block_identity} holds by virtue of a property of Moore-Penrose inverse, $\bC^\dagger = \bC^\dagger\bC\bC^\dagger$, whereas \eqref{aclasse} follows from condition
\eqref{schur_condition3}. Then, the block matrix in \eqref{block_matrix} can be expressed as 
\begin{align}
    \begin{bmatrix}
        \bI & -\bB\bC^\dagger\\
        \zero & \bI
    \end{bmatrix}^{-1}
    \begin{bmatrix}
        \bA-\bB\bC^\dagger\bB^\top & \zero \\
        \zero & \bC
    \end{bmatrix}
    \begin{bmatrix}
        \bI & \zero\\
        -\bC^\dagger\bB^\top & \bI
    \end{bmatrix}^{-1}, \nonumber
\end{align}
which is positive-semidefinite because the central matrix is positive-semidefinite from conditions \eqref{schur_condition1} and \eqref{schur_condition2}, while the leading and trailing matrices are invertible as they are triangular with ones on the diagonals.

Turning to the `only if' part, from the identity above, the block matrix in \eqref{block_identity} is positive-semidefinite, which implies the positive-semidefiniteness of the submatrices $\bC$ and $\bA-\bB\bC^\dagger\bB^\top$. It then remains to prove that the positive-semidefiniteness of \eqref{block_matrix} implies \eqref{rangeCondition}. Letting $\bx\in\operatorname{null}\bC$,
\begin{align}
    &\begin{bmatrix}
    \zero^\top & \bx^\top
    \end{bmatrix}
    \begin{bmatrix}
        \bA &\bB\\
        \bB^\top & \bC
    \end{bmatrix}
    \begin{bmatrix}
    \zero \\ \bx
    \end{bmatrix} = 
    \bx^\top\bC\bx = 0\\
    & \label{paella} \Rightarrow 
    \begin{bmatrix}
        \bA &\bB\\
        \bB^\top & \bC
    \end{bmatrix}^{\frac{1}{2}}
    \begin{bmatrix}
    \zero \\ \bx
    \end{bmatrix} = \zero\\
    &\Rightarrow 
    \begin{bmatrix}
    \bB\bx \\ \bC\bx
    \end{bmatrix} =\begin{bmatrix}
        \bA &\bB\\
        \bB^\top & \bC
    \end{bmatrix}
    \begin{bmatrix}
    \zero \\ \bx
    \end{bmatrix}\\
    &\qquad\qquad\,= \begin{bmatrix}
        \bA &\bB\\
        \bB^\top & \bC
    \end{bmatrix}^{\frac{1}{2}}\cdot
    \begin{bmatrix}
        \bA &\bB\\
        \bB^\top & \bC
    \end{bmatrix}^{\frac{1}{2}}
    \begin{bmatrix}
    \zero \\ \bx
    \end{bmatrix}  = \zero\nonumber\\
    &\Rightarrow \bB\bx = \zero\\
    &\Rightarrow \bx \in \operatorname{null}\bB,
\end{align}
where \eqref{paella} follows from the positive-semidefiniteness of the block matrix in \eqref{block_matrix}. This proves $\operatorname{null}\bC \subset \operatorname{null}\bB$, which is equivalent to \eqref{rangeCondition}.

\end{proof}

Also noteworthy is that the upper-left block of \eqref{block_identity},
\begin{align}
    \bA-\bB\bC^\dagger\bB^\top
    =\begin{bmatrix}
        \bI & -\bB\bC^\dagger
    \end{bmatrix}
    \begin{bmatrix}
        \bA &\bB\\
        \bB^\top & \bC
    \end{bmatrix}
    \begin{bmatrix}
        \bI\\
        -\bC^\dagger\bB^\top
    \end{bmatrix}, \nonumber
\end{align}
proves by itself a weaker version of the Schur complement lemma, namely that
\begin{align}
    \begin{bmatrix}
        \bA &\bB\\
        \bB^\top & \bC
    \end{bmatrix} \geq \zero \Rightarrow \bA-\bB\bC^\dagger\bB^\top \geq \zero.
    \label{claustre}
\end{align}
This derivation can be found in \cite[Lemma 1]{gorman1990lower},
extended to
\begin{align}
    \begin{bmatrix}
        \bI & -\bS^\top
    \end{bmatrix}
    \begin{bmatrix}
        \bA &\bB\\
        \bB^\top & \bC
    \end{bmatrix}
    \begin{bmatrix}
        \bI\\
        -\bS
    \end{bmatrix} \label{stoica_and_marzetta}
\end{align}
in \cite{stoica2001parameter}, of which \eqref{claustre} is a special case for $\bS = \bB\bC^\dagger$.
Actually, 
the proof in \cite{stoica2001parameter} can be understood as a proof of the forward part of the Schur complement lemma.

Consider the spectral decomposition
\begin{align}
    \bC = \bU\bLambda\bU^\top = 
    \begin{bmatrix}
        \bU_1 & \bU_2
    \end{bmatrix}
    \begin{bmatrix}
        \bLambda_1 & \zero\\
        \zero & \zero
    \end{bmatrix}
    \begin{bmatrix}
        \bU_1^\top \\ \bU_2^\top
    \end{bmatrix}, 
    \label{spectral_decomposition}
\end{align}
where $\bLambda_1$ is a diagonal matrix with positive entries.
Letting
\begin{align}
    \begin{bmatrix}
    \bB_1 & \bB_2    
    \end{bmatrix}
     = \bB\bU \qquad\quad
    \begin{bmatrix}
    \bS_1 \\ \bS_2    
    \end{bmatrix}
     = \bU^\top\bS,
\end{align}
\eqref{stoica_and_marzetta} becomes
\begin{align}
    & \!\!\!\!\! \begin{bmatrix}
        \bI & -\bS^\top
    \end{bmatrix}
    \begin{bmatrix}
        \bA &\bB\\
        \bB^\top & \bU\bLambda\bU^\top
    \end{bmatrix}
    \begin{bmatrix}
        \bI\\
        -\bS
    \end{bmatrix} \nonumber \\
    & \quad =\begin{bmatrix}
        \bI & -\bS^\top\bU
    \end{bmatrix}
    \begin{bmatrix}
        \bA &\bB\bU\\
        \bU^\top\bB^\top & \bLambda
    \end{bmatrix}
    \begin{bmatrix}
        \bI\\
        -\bU^\top\bS
    \end{bmatrix}\\
    & \quad =\begin{bmatrix}
        \bI & -\bS_1^\top & -\bS_2^\top
    \end{bmatrix}
    \begin{bmatrix}
        \bA &\bB_1 & \bB_2\\
        \bB_1^\top & \bLambda_1 & \zero\\
        \bB_2^\top & \zero & \zero
    \end{bmatrix}
    \begin{bmatrix}
        \bI\\
        -\bS_1\\
        -\bS_2
    \end{bmatrix}\\
    & \quad=\begin{bmatrix}
        \bI & -\bS_1^\top
    \end{bmatrix}
    \begin{bmatrix}
        \zero &\bB_1\\
        \bB_1^\top & \bLambda_1
    \end{bmatrix}
    \begin{bmatrix}
        \bI\\
        -\bS_1
    \end{bmatrix}  \nonumber \\
    & \quad\quad + \bA - \bB_2\bS_2 - \bS_2^\top\bB_2^\top.
\end{align}
Completing the square turns the above into
\begin{align}
    & \begin{bmatrix}
        \bI & -\bS_1^\top
    \end{bmatrix}
    \begin{bmatrix}
        \bB_1\bLambda_1^{-1}\bB_1^\top &\bB_1\\
        \bB_1^\top & \bLambda_1
    \end{bmatrix}
    \begin{bmatrix}
        \bI\\
        -\bS_1
    \end{bmatrix} \nonumber \\
    &\qquad\qquad + \bA - \bB_2\bS_2 - \bS_2^\top\bB_2^\top - \bB_1\bLambda_1^{-1}\bB_1^\top, \nonumber
\end{align}
with the leading term further simplified into
\begin{align}
    &\begin{bmatrix}
        \bI & -\bS_1^\top
    \end{bmatrix}
    \begin{bmatrix}
        \bB_1\bLambda_1^{-1}\\
        \bI
    \end{bmatrix}
    \bLambda_1
    \begin{bmatrix}
        \bLambda_1^{-1}\bB_1^\top &
        \bI
    \end{bmatrix}
    \begin{bmatrix}
        \bI\\
        -\bS_1
    \end{bmatrix}\\
    &\qquad =(\bLambda_1^{-1}\bB_1^\top-\bS_1)^\top\bLambda_1 \,
    (\bLambda_1^{-1}\bB_1^\top-\bS_1).
\end{align}
If the block matrix in \eqref{block_matrix} is positive-semidefinite, then
\begin{align}
    &(\bLambda_1^{-1}\bB_1^\top-\bS_1)^\top\bLambda_1 \,
    (\bLambda_1^{-1}\bB_1^\top-\bS_1)\label{completing_the_square}\\
    &\qquad\qquad\qquad + \bA - \bB_2\bS_2 - \bS_2^\top\bB_2^\top - \bB_1\bLambda_1^{-1}\bB_1^\top \geq \zero. \nonumber
\end{align}
As $\bLambda_1$ has positive diagonal entries, the above inequality is tightest for $\bS_1 = \bLambda_1^{-1}\bB_1^\top$, resulting in
\begin{align}
    \bA \geq \bB_2\bS_2 + \bS_2^\top\bB_2^\top + \bB_1\bLambda_1^{-1}\bB_1^\top \label{quadratic_form_simplified}
\end{align}
with the latter term further simplifying into
\begin{align}
    \bB_1\bLambda_1^{-1}\bB_1^\top &= \bB\bU_1\bLambda_1^{-1}\bU_1^\top\bB^\top\\
    &=\bB\bC^\dagger\bB^\top.
\end{align}
As \eqref{quadratic_form_simplified} holds for every $\bS_2$, it holds for $\bS_2 = \zero$, which proves that $\bA\geq \bB\bC^\dagger\bB^\top$ and thus condition \eqref{schur_condition2}.

To prove condition \eqref{schur_condition3}, in turn, let $\bS_2 = \frac{t}{2}\bB_2^\top$, whereby \eqref{quadratic_form_simplified} becomes
\begin{align}
    \bB_2\bB_2^\top\leq \frac{1}{t} \left(\bA - \bB\bC^\dagger\bB^\top \right),
\end{align}
which, letting $t\rightarrow \infty$ and given the positive-semidefiniteness of $\bB_2\bB_2^\top$, gives $\bB_2\bB_2^\top = \zero$.
From
\begin{align}
    \bB_2\bB_2^\top &= \bB\bU_2\bU_2^\top\bB^\top\\
    &= \bB \left( \bI-\bU_1\bU_1^\top \right) \bB^\top\\
    &= \bB \left( \bI-\bC\bC^\dagger \right) \bB^\top\\
    &= \left(\bB \left( \bI-\bC\bC^\dagger \right)\right) \left( \left(\bI-\bC\bC^\dagger \right)\bB^\top\right),
\end{align}
it follows that, as desired,
\begin{align}
    \left( \bI-\bC\bC^\dagger \right)\bB^\top = \zero .
\end{align}

The proof of the converse part is not available in \cite{stoica2001parameter}, but can be obtained in a similar fashion as follows.
Under \eqref{schur_condition1}--\eqref{schur_condition3},
\begin{align}
    &\bA-\bB\bC^\dagger\bB^\top + (\bS_1 - \bLambda_1^{-1}\bB_1^\top)^\top\bLambda_1(\bS_1 - \bLambda_1^{-1}\bB_1^\top) \geq \zero \nonumber
\end{align}
for any $\bS$. In conjunction with
\begin{align}
    \bB_2\bS_2 &= \bB\bU_2\bU_2^\top\bS\\
    &= \bB(\bI-\bU_1\bU_1^\top)\bS\\
    &= \bB(\bI-\bC\bC^\dagger)\bS=\zero,
\end{align}
it can be shown that \eqref{completing_the_square} holds for any $\bS$, ensuring the positive-semidefiniteness of \eqref{stoica_and_marzetta} for any $\bS$. For any vector
\begin{align}
    \bw = \begin{bmatrix}
        \bw_1\\
        \bw_2
    \end{bmatrix}
\end{align}
where $\bw_1 \neq \zero$,
the matrix $\bS = -\frac{\bw_2\bw_1^\top}{\|\bw_1\|^2}$ satisfies
\begin{align}
    \bw \in \col \begin{bmatrix}
        \bI \\ -\bS
    \end{bmatrix}
\end{align}
as
\begin{align}
    \bw = \begin{bmatrix}
        \bI \\ \frac{\bw_2\bw_1^\top}{\|\bw_1\|^2}
    \end{bmatrix}\bw_1.
\end{align}
Therefore,
\begin{align}
    \bw^\top \begin{bmatrix}
        \bA &\bB\\
        \bB^\top & \bC
    \end{bmatrix} \bw \geq 0
\end{align}
for every $\bw$ for which $\bw_1 \neq \zero$. Continuity of the left-hand side extends the validity of the statement to $\bbR^k$.

\section{}
\label{app:unconstrained_bound}
Consider the random vector
\begin{align}
\begin{bmatrix}
    \hat{\btheta}-\bbE[\hat{\btheta}]\\
    \big(\frac{\partial L}{\partial \btheta}\big)^{\!\!\top}
\end{bmatrix},
\end{align}
whose covariance is
\begin{align}
    \begin{bmatrix}
        \bbE\big[(\hat{\btheta}-\bbE[\hat{\btheta}])(\hat{\btheta}-\bbE[\hat{\btheta}])^{\top}\big] &  \bbE\Big[(\hat{\btheta}-\bbE[\hat{\btheta}]) \frac{\partial L}{\partial \btheta} \Big] \\
        \bbE\Big[\big(\frac{\partial L}{\partial \btheta}\big)^{\!\top} (\hat{\btheta}-\bbE[\hat{\btheta}])^\top\Big] &
        \bbE\Big[\big(\frac{\partial L}{\partial \btheta}\big)^{\!\top}
        \frac{\partial L}{\partial \btheta}
        \Big]
    \end{bmatrix}. \label{covariance_unconstrained}
\end{align}
Now consider two identities that hold for $\btheta \in \Theta$, namely
\begin{align}
    1 = \int p(\by;\btheta) \, d\by
    \label{identity1}
\end{align}    
    and
\begin{align}
\label{identity2}
    \btheta + \bb = \int \hat{\btheta} \, p(\by;\btheta) \, d\by.
\end{align}
Differentiating both sides of each identity with respect to $\btheta$ gives, respectively,
\begin{align}
    \zero^\top &  = \int \frac{\partial p(\by;\btheta)}{\partial \btheta} d\by \\
    & =\int \frac{1}{p(\by;\btheta)}\frac{\partial p(\by;\btheta)}{\partial \btheta} \cdot  p(\by;\btheta) \, d\by\\
    & = \bbE\bigg[\frac{\partial L}{\partial \btheta}\bigg]
\end{align}
and 
\begin{align}
    \bI + \frac{\partial \bb}{\partial \btheta} & = \int \hat{\btheta}\frac{\partial p(\by;\btheta)}{\partial \btheta} d\by\\
    &= \bbE\bigg[\hat{\btheta} \frac{\partial L}{\partial \btheta}\bigg],
\end{align}
which turn \eqref{covariance_unconstrained} into
\begin{align}
    \begin{bmatrix}
        \bbE\big[(\hat{\btheta}-\bbE[\hat{\btheta}])(\hat{\btheta}-\bbE[\hat{\btheta}])^{\top}\big] &  \bI+\frac{\partial \bb}{\partial \btheta} \\
        \big(\bI+\frac{\partial \bb}{\partial \btheta}\big)^{\!\top} &
        \bbE\Big[\big(\frac{\partial L}{\partial \btheta}\big)^{\!\top}
        \frac{\partial L}{\partial \btheta}
        \Big]
    \end{bmatrix}.
\end{align}
The application of the Schur complement lemma, given in App. \ref{app:schur}, concludes the proof.

\section{}
\label{app:constrained_bound}

Consider the random vector
\begin{align}
\begin{bmatrix}
    \hat{\btheta}-\bbE[\hat{\btheta}]\\
    \frac{\partial L}{\partial \bu_1} \\ \vdots \\ \frac{\partial L}{\partial \bu_\ell}
\end{bmatrix}
\end{align}
with covariance
\begin{align}
    \begin{bmatrix}
        \bbE\big[(\hat{\btheta}-\bbE[\hat{\btheta}])(\hat{\btheta}-\bbE[\hat{\btheta}])^{\!\top} \big] &  \bbE\Big[(\hat{\btheta}-\bbE[\hat{\btheta}]) \! \begin{bmatrix}
        \frac{\partial L}{\partial \bu_1} & \!\cdots\! & \frac{\partial L}{\partial \bu_\ell}
        \end{bmatrix} \! \Big] \\
        \bbE\left[\begin{bmatrix}
        \frac{\partial L}{\partial \bu_1} \\ \vdots \\ \frac{\partial L}{\partial \bu_\ell}
        \end{bmatrix} (\hat{\btheta}-\bbE[\hat{\btheta}])^{\!\top}\right] &
        \bbE\left[
        \begin{bmatrix}
            \frac{\partial L}{\partial \bu_1} \\ \vdots \\ \frac{\partial L}{\partial \bu_\ell}
        \end{bmatrix}
        \begin{bmatrix}
            \frac{\partial L}{\partial \bu_1} & \cdots & \frac{\partial L}{\partial \bu_\ell}
        \end{bmatrix}
        \right]
    \end{bmatrix} \label{covariance}
\end{align}
and the identities in \eqref{identity1} and \eqref{identity2}.
Taking directional derivatives of both sides of each identity with respect to $\bu_i\in\cT_{\Theta}(\btheta)$ gives, respectively,
\begin{align}
    0  = \int \frac{\partial p(\by;\btheta)}{\partial \bu_i} d\by = \bbE\bigg[\frac{\partial L}{\partial \bu_i}\bigg]
\end{align}
and
\begin{align}
    \bu_i + \frac{\partial \bb}{\partial \bu_i} & = \int \hat{\btheta}\frac{\partial p(\by;\btheta)}{\partial \bu_i} d\by = \bbE\bigg[\hat{\btheta} \frac{\partial L}{\partial \bu_i}\bigg].
\end{align}
These turn \eqref{covariance} into
\begin{align}
    \begin{bmatrix}
        \bbE[(\hat{\btheta}-\bbE[\hat{\btheta}])(\hat{\btheta}-\bbE[\hat{\btheta}])^{\!\top}] &  \begin{bmatrix}
            \bu_1 + \frac{\partial \bb}{\partial \bu_1} & \cdots & \bu_\ell + \frac{\partial \bb}{\partial \bu_\ell}
        \end{bmatrix} \\
        \begin{bmatrix}
            \big(\bu_1 + \frac{\partial \bb}{\partial \bu_1}\big)^{\!\top} \\ \vdots \\ \big(\bu_\ell + \frac{\partial \bb}{\partial \bu_\ell}\big)^{\!\top}
        \end{bmatrix} &
        \bbE\left[
        \begin{bmatrix}
            \frac{\partial L}{\partial \bu_1} \\ \vdots \\ \frac{\partial L}{\partial \bu_\ell}
        \end{bmatrix}
        \begin{bmatrix}
            \frac{\partial L}{\partial \bu_1} & \cdots & \frac{\partial L}{\partial \bu_\ell}
        \end{bmatrix}
        \right]
    \end{bmatrix},
\end{align}
which can be further simplified into
\begin{align}
    \begin{bmatrix}
        \bbE[(\hat{\btheta}-\bbE[\hat{\btheta}])(\hat{\btheta}-\bbE[\hat{\btheta}])^{\!\top}] & \quad  \big(\bI + \frac{\partial \bb}{\partial \btheta}\big)\bU \\
         \bU^\top\big(\bI + \frac{\partial \bb}{\partial \btheta}\big)^{\!\!\top} &
        \quad \bU^\top\bJ\bU
    \end{bmatrix} \label{covariance_simple}
\end{align}
given that
\begin{align}
    \frac{\partial \bb}{\partial \bu_i} = \frac{\partial \bb}{\partial \btheta} \bu_i \qquad\quad \frac{\partial L}{\partial \bu_i} = \frac{\partial L}{\partial \btheta} \bu_i. \label{chain_rule}
\end{align}
The application of the Schur complement lemma, given in App. \ref{app:schur}, concludes the proof.

\section{}
\label{ICREA}

The implication (C)$\;\Rightarrow\;$(A) is trivial, and what needs to be proved is (A)$\;\Rightarrow\;$(B) and (B)$\;\Rightarrow\;$(C).
Let us begin with (A)$\;\Rightarrow\;$(B).
As per Lemma~\ref{lemma:existence_V}, there exists a matrix $\bV = [\bv_1\, \cdots \, \bv_d]$ whose columns are in $\cT_\Theta(\btheta)$ and whose column space is $\vspan \cT_\Theta(\btheta)$.
As its column space equals $\vspan \cT_\Theta(\btheta)$, we have
\begin{align}
    \bPi = \bV \bV^\dagger = (\bV^\dagger)^\top\bV^\top.
\end{align}
As each column of $\bV$ is in $\cT_\Theta(\btheta)$, we further have---recall \eqref{axler}---that
\begin{align}
    \col \! \left(\bV^\top\bigg(\bI + \frac{\partial \bb}{\partial \btheta}\bigg)^{\!\!\top}\right) \subset \col \! \big( \bV^\top\bJ^{\frac{1}{2}} \big).
\end{align}
Multiplying both sides by $(\bV^\dagger)^\top$ gives
\begin{align}
    \col \! \left(\bPi \, \bigg(\bI + \frac{\partial \bb}{\partial \btheta}\bigg)^{\!\!\top}\right) \subset \col \! \big(\bPi\bJ^{\frac{1}{2}} \big),
\end{align}
and applying \eqref{axler} yields \eqref{range_ccrb_projection}.
As of (B)$\,\Rightarrow$(C), it holds because, for any $\bU$ whose columns are in $\cT_\Theta(\btheta)$,
\begin{align}
    \col\!\left(\bU^\top\bigg(\bI + \frac{\partial \bb}{\partial \btheta}\bigg)^{\!\!\top}\right) & = \col\!\left(\bU^\top\bPi \, \bigg(\bI + \frac{\partial \bb}{\partial \btheta}\bigg)^{\!\!\top}\right) \nonumber \\
    &\subset \col \! \big( \bU^\top\bPi\bJ^{\frac{1}{2}} \big) \\
    &=\col \! \big( \bU^\top\bJ^{\frac{1}{2}} \big), 
\end{align}
where the first and the last steps follow from $\bPi \bU = \bU$ while the middle step follows from the initial assumption.



\section{}
\label{bombers}
Every spanning set contains a basis in a finite-dimensional vector space \cite[Thm. 3.5]{nering1970linear}. 
Let us pick $\bv_1$ from $\cT_\Theta(\btheta)$. For $i\geq 2$, let us choose $\bv_i\in\cT_\Theta(\btheta)\setminus\vspan\{\bv_1, \ldots, \bv_{i-1}\}$ until the set becomes empty. The procedure must terminate within finite steps 
as $\cT_\Theta(\btheta)$ is a subset of a finite-dimensional vector space. By construction, $\{\bv_1,\ldots,\bv_d\}$ is linearly independent, and it spans $\vspan \cT_\Theta(\btheta)$ from
\begin{align}
    &\{\bv_1,\ldots,\bv_d\} \subset \cT_\Theta(\btheta)\\
    &\qquad \Rightarrow \vspan\{\bv_1,\ldots,\bv_d\} \subset \vspan\cT_\Theta(\btheta)
\end{align}
and
\begin{align}
    &\cT_\Theta(\btheta)\setminus\vspan\{\bv_1, \ldots, \bv_d\} = \emptyset\\
    &\qquad\Leftrightarrow \cT_\Theta(\btheta) \subset \vspan\{\bv_1, \ldots, \bv_d\}\\
    &\qquad\Rightarrow \vspan \cT_\Theta(\btheta) \subset \vspan\{\bv_1, \ldots, \bv_d\}.
\end{align}

\section{}
\label{DraR}

The proof entails applying the Schur complement lemma to the positive-semidefinite matrix
\begin{align*}
   &\begin{bmatrix}
       \bW_2(\bW_2^\top\bJ\bW_2)^{-1}\bW_2^\top & \bW_1\\
       \bW_1^\top & \bW_1^\top\bJ\bW_1
   \end{bmatrix}\\
   &=\begin{bmatrix}
       \bW_2(\bW_2^\top\bJ\bW_2)^{-1}\bW_2^\top & \bW_2\bW_2^\dagger\bW_1\\
       \bW_1^\top(\bW_2^\dagger)^\top\bW_2^\top & \bW_1^\top(\bW_2^\dagger)^\top\bW_2^\top\bJ\bW_2\bW_2^\dagger\bW_1
   \end{bmatrix}\!\! \nonumber \\
   &=\begin{bmatrix}
       \bW_2\\ \bW_1^\top(\bW_2^\dagger)^\top(\bW_2^\top\bJ\bW_2)
   \end{bmatrix}\\
   &\qquad\qquad\qquad \cdot (\bW_2^\top\bJ\bW_2)^{-1}
   \begin{bmatrix}
       \bW_2^\top & (\bW_2^\top\bJ\bW_2)\bU_2^\dagger\bW_1
   \end{bmatrix}, \nonumber
\end{align*}
where $\bW_2\bW_2^\dagger\bW_1 = \bW_1$ was used. An alternative derivation of this result can be found in \cite[App. B]{tang2011lower}.

\section{}
\label{SOMMA}
Using the following facts, one can obtain the series of relationships in \eqref{ccrb_monotonicity_proof_start}--\eqref{ccrb_monotonicity_proof_end}.
\begin{figure*}
\begin{align}
    &\bigg(\bI + \frac{\partial \bb}{\partial \btheta}\bigg)\bW_1(\bW_1^\top \bJ \bW_1)^{\dagger}\bW_1^\top\bigg(\bI + \frac{\partial \bb}{\partial \btheta}\bigg)^{\!\!\top} \nonumber \\
    &=\bigg(\bI + \frac{\partial \bb}{\partial \btheta}\bigg)\bW_2\bW_2^\dagger\bW_1(\bW_1^\top \bJ \bW_1)^{\dagger}\bW_1^\top(\bW_2^\dagger)^\top\bW_2^\top\bigg(\bI + \frac{\partial \bb}{\partial \btheta}\bigg)^{\!\!\top} &&\text{F1} \label{ccrb_monotonicity_proof_start} \\
    &=\bigg(\bI + \frac{\partial \bb}{\partial \btheta}\bigg)\bW_2(\bJ^{\frac{1}{2}}\bW_2)^\dagger\bJ^{\frac{1}{2}}\bW_2\bW_2^\dagger\bW_1(\bW_1^\top \bJ \bW_1)^{\dagger}\bW_1^\top(\bW_2^\dagger)^\top\bW_2^\top \bJ^{\frac{1}{2}}(\bW_2^\top \bJ^{\frac{1}{2}})^{\dagger}\bW_2^\top\bigg(\bI + \frac{\partial \bb}{\partial \btheta}\bigg)^{\!\!\top} &&\text{F2}\\
    &=\bigg(\bI + \frac{\partial \bb}{\partial \btheta}\bigg)\bW_2(\bJ^{\frac{1}{2}}\bW_2)^\dagger\bJ^{\frac{1}{2}}\bW_1(\bW_1^\top \bJ \bW_1)^{\dagger}\bW_1^\top \bJ^{\frac{1}{2}}(\bW_2^\top \bJ^{\frac{1}{2}})^{\dagger}\bW_2^\top\bigg(\bI + \frac{\partial \bb}{\partial \btheta}\bigg)^{\!\!\top} &&\text{F1}\\
    &\leq \bigg(\bI + \frac{\partial \bb}{\partial \btheta}\bigg)\bW_2(\bJ^{\frac{1}{2}}\bW_2)^\dagger(\bW_2^\top \bJ^{\frac{1}{2}})^{\dagger}\bW_2^\top\bigg(\bI + \frac{\partial \bb}{\partial \btheta}\bigg)^{\!\!\top} &&\text{F3}\\
    &= \bigg(\bI + \frac{\partial \bb}{\partial \btheta}\bigg)\bW_2(\bW_2^\top\bJ\bW_2)^\dagger\bW_2^\top\bJ^{\frac{1}{2}}(\bW_2^\top \bJ^{\frac{1}{2}})^{\dagger}\bW_2^\top\bigg(\bI + \frac{\partial \bb}{\partial \btheta}\bigg)^{\!\!\top} &&\text{F4}\\
    &= \bigg(\bI + \frac{\partial \bb}{\partial \btheta}\bigg)\bW_2(\bW_2^\top\bJ\bW_2)^\dagger\bW_2^\top\bigg(\bI + \frac{\partial \bb}{\partial \btheta}\bigg)^{\!\!\top} &&\text{F2} \label{ccrb_monotonicity_proof_end}
\end{align}
\hrulefill
\end{figure*}
\begin{enumerate}
    \item[F1.] $\bW_2\bW_2^\dagger\bW_1 = \bW_1$ owing to $\col \bW_1 \subset \col \bW_2$.
    \item[F2.] $\bW_2^\top\big(\bI + \frac{\partial \bb}{\partial \btheta}\big)^{\!\top} = (\bW_2^\top\bJ^{\frac{1}{2}})(\bW_2^\top\bJ^{\frac{1}{2}})^\dagger\bW_2^\top\big(\bI + \frac{\partial \bb}{\partial \btheta}\big)^{\!\top}$ from $\col \bW_2^\top\big(\bI + \frac{\partial \bb}{\partial \btheta}\big)^{\!\top} \subset \col \bW_2^\top\bJ^{\frac{1}{2}}$ (recall Thm.~\ref{theorem:range_equivalence}).
    \item[F3.] $\bJ^{\frac{1}{2}}\bW_1(\bW_1^\top \bJ \bW_1)^{\dagger}\bW_1^\top \bJ^{\frac{1}{2}} = \bJ^{\frac{1}{2}}\bW_1 (\bJ^{\frac{1}{2}}\bW_1)^\dagger \leq \bI$ from the fact that the left-hand side is a projection matrix onto $\col(\bJ^{\frac{1}{2}}\bW_1)$.
    \item[F4.] $(\bJ^{\frac{1}{2}}\bW_2)^\dagger = (\bW_2^\top\bJ\bW_2)^\dagger\bW_2^\top\bJ^{\frac{1}{2}}$ from \eqref{moore_penrose_property}.
\end{enumerate}

\section{}
\label{Chelsea}

Let $\btheta^*$ be an extreme point of $\Theta$ and let $\hat{\btheta}(\by)$ be an unbiased estimator. With $\by$ generated from $p(\by;\btheta^*)$, given the unbiasedness,
\begin{align}
    \bbE_{\by \sim p(\by;\btheta^*)}\big[\hat{\btheta}(\by)\big] = \btheta^*.
\end{align}
The left-hand side is, by definition of expectation, a convex combination of the points in $\Theta$. Meanwhile, the right-hand side, by definition of extreme point, cannot be written as a nontrivial convex combination of the points in $\Theta$. Combining these two observations, 
it can be concluded that
\begin{align}
    \hat{\btheta}(\by) = \btheta^* \label{constant_estimator}
\end{align}
with probability one.
If there is another $\btheta\in\Theta$ such that the support of $p(\by;\btheta)$ is included in that of $p(\by;\btheta^*)$, one can surmise \eqref{constant_estimator} for $\by\sim p(\by;\btheta)$. This gives 
\begin{align}
    \bbE_{\by\sim p(\by;\btheta)}\big[\hat{\btheta}(\by)\big] = \btheta^*, \label{contradiction}
\end{align}
which contradicts the unbiased assumption.

\bibliographystyle{IEEEtran}
\bibliography{ref}

\end{document}

%% file: input.tex
\def\zero{{\boldsymbol{0}}}

\def\bb0{{\mathbb{0}}}

\def\bb{{\boldsymbol{b}}}

\def\bee{{\boldsymbol{e}}}

\def\bu{{\boldsymbol{u}}}
\def\bv{{\boldsymbol{v}}}
\def\bw{{\boldsymbol{w}}}
\def\bx{{\boldsymbol{x}}}
\def\by{{\boldsymbol{y}}}

\def\b0{{\boldsymbol{0}}}

\def\bA{{\boldsymbol{A}}}
\def\bB{{\boldsymbol{B}}}
\def\bC{{\boldsymbol{C}}}

\def\bF{{\boldsymbol{F}}}

\def\bI{{\boldsymbol{I}}}
\def\bJ{{\boldsymbol{J}}}

\def\bM{{\boldsymbol{M}}}

\def\bS{{\boldsymbol{S}}}

\def\bU{{\boldsymbol{U}}}
\def\bV{{\boldsymbol{V}}}
\def\bW{{\boldsymbol{W}}}
\def\bX{{\boldsymbol{X}}}

\def\b{{\mathrm{b}}}

\def\r0{{\mathbf{0}}}

\def\bbE{{\mathbb{E}}}

\def\bbR{{\mathbb{R}}}

\def\cN{\mathcal{N}}

\def\cT{\mathcal{T}}

\def\btheta{\bm \theta}

\def\brho{\bm \rho}

\def\bLambda{\bm \Lambda}
\def\bSigma{\bm \Sigma}

\def\bPi{\bm \Pi}

\def\bsf0{{\bm{\mathsf{0}}}}

\def\N0{{N_{\mathrm{0}}}}

\def\bsf{{\boldsymbol{s}_\mathrm{f}}}

\newcommand{\be}{\begin{equation}}
\newcommand{\ee}{\end{equation}}
\newcommand{\bal}{\begin{align}}
\newcommand{\eal}{\end{align}}
\def\tr {{\rm tr}}

